\documentclass[nohyperref]{article}

\usepackage{microtype}
\usepackage{graphicx}
\usepackage{subfigure}
\usepackage{booktabs}
\usepackage{hyperref}

\usepackage[accepted]{icml2023}

\usepackage{amssymb,amsmath,amsthm,amsfonts}
\usepackage{bm}
\usepackage{textgreek}
\usepackage{mathtools}
\usepackage{enumitem}
\usepackage{authblk}
\usepackage{graphicx}
\usepackage[font=small]{caption}
\usepackage{float}
\usepackage{physics}
\usepackage{footnote}
\usepackage{xcolor}
\usepackage{mathrsfs}
\usepackage{bbm}
\usepackage{enumitem}

\usepackage[capitalize,noabbrev]{cleveref}

\theoremstyle{plain}
\newtheorem{theorem}{Theorem}[section]
\newtheorem{proposition}[theorem]{Proposition}
\newtheorem{lemma}[theorem]{Lemma}

\theoremstyle{definition}
\newtheorem{definition}[theorem]{Definition}
\newtheorem{assumption}[theorem]{Assumption}
\newtheorem{problem}[theorem]{Problem}
\theoremstyle{remark}
\newtheorem{remark}[theorem]{Remark}

\newcommand{\eqn}[1]{(\ref{eqn:#1})}

\newcommand{\thm}[1]{\hyperref[thm:#1]{Theorem~\ref*{thm:#1}}}
\newcommand{\cor}[1]{\hyperref[cor:#1]{Corollary~\ref*{cor:#1}}}
\newcommand{\defn}[1]{\hyperref[defn:#1]{Definition~\ref*{defn:#1}}}
\newcommand{\lem}[1]{\hyperref[lem:#1]{Lemma~\ref*{lem:#1}}}
\newcommand{\prop}[1]{\hyperref[prop:#1]{Proposition~\ref*{prop:#1}}}
\newcommand{\prob}[1]{\hyperref[prob:#1]{Problem~\ref*{prob:#1}}}
\newcommand{\assum}[1]{\hyperref[assum:#1]{Assumption~\ref*{assum:#1}}}
\newcommand{\fig}[1]{\hyperref[fig:#1]{Figure~\ref*{fig:#1}}}
\newcommand{\tab}[1]{\hyperref[tab:#1]{Table~\ref*{tab:#1}}}
\newcommand{\algo}[1]{\hyperref[algo:#1]{Algorithm~\ref*{algo:#1}}}
\renewcommand{\sec}[1]{\hyperref[sec:#1]{Section~\ref*{sec:#1}}}
\newcommand{\append}[1]{\hyperref[append:#1]{Appendix~\ref*{append:#1}}}
\newcommand{\fac}[1]{\hyperref[fac:#1]{Fact~\ref*{fac:#1}}}
\newcommand{\lin}[1]{\hyperref[lin:#1]{Line~\ref*{lin:#1}}}

\usepackage[textsize=tiny]{todonotes}

\icmltitlerunning{Quantum Lower Bounds for Finding Stationary Points of Nonconvex Functions}

\def\>{\rangle}
\def\<{\langle}

\newcommand{\vect}[1]{\ensuremath{\mathbf{#1}}}
\newcommand{\x}{\ensuremath{\mathbf{x}}}

\newcommand{\N}{\mathbb{N}}

\newcommand{\R}{\mathbb{R}}

\newcommand{\E}{\mathbb{E}}

\DeclareMathOperator{\poly}{poly}
\DeclareMathOperator{\spn}{span}

\DeclareMathOperator{\supp}{supp}
\DeclareMathOperator{\quan}{quan}
\DeclareMathOperator{\sto}{stoc}
\DeclareMathOperator{\prog}{prog}

\renewcommand{\d}{\mathrm{d}}
\renewcommand{\x}{\vect{x}}

\newcommand{\y}{\vect{y}}
\newcommand{\z}{\vect{z}}
\newcommand{\g}{\vect{g}}
\newcommand{\m}{\vect{m}}
\renewcommand{\u}{\vect{u}}
\newcommand{\0}{\mathbf{0}}
\newcommand{\tildef}{\tilde{f}_{T;U}}
\newcommand{\hatf}{\hat{f}_{T;U}}
\newcommand{\tildeOp}{\widetilde{O}^{(p)}}
\newcommand{\tildeO}{\widetilde{O}}

\def\:{\hbox{\bf:}}

\let\oldnl\nl
\newcommand{\nonl}{\renewcommand{\nl}{\let\nl\oldnl}}

\begin{document}

\twocolumn[
\icmltitle{Quantum Lower Bounds for Finding Stationary Points of Nonconvex Functions}

\begin{icmlauthorlist}
\icmlauthor{Chenyi Zhang}{1,2}
\icmlauthor{Tongyang Li}{3,4}
\end{icmlauthorlist}

\icmlaffiliation{1}{Computer Science Department, Stanford University}
\icmlaffiliation{2}{Institute for Interdisciplinary Information Sciences, Tsinghua University}
\icmlaffiliation{3}{Center on Frontiers of Computing Studies, Peking University}
\icmlaffiliation{4}{School of Computer Science, Peking University}

\icmlcorrespondingauthor{Tongyang Li}{tongyangli@pku.edu.cn}

\icmlkeywords{Machine Learning, ICML}

\vskip 0.3in
]

\printAffiliationsAndNotice{}

\begin{abstract}
Quantum computing is an emerging technology that has been rapidly advancing in the past decades. In this paper, we conduct a systematic study of quantum lower bounds on finding $\epsilon$-approximate stationary points of nonconvex functions, and we consider the following two important settings: 1) having access to $p$-th order derivatives; or 2) having access to stochastic gradients.
The classical query lower bounds are $\Omega\big(\epsilon^{-\frac{1+p}{p}}\big)$ regarding the first setting and $\Omega(\epsilon^{-4})$ regarding the second setting (or $\Omega(\epsilon^{-3})$ if the stochastic gradient function is mean-squared smooth). In this paper, we extend all these classical lower bounds to the quantum setting. They match the classical algorithmic results respectively, demonstrating that there is no quantum speedup for finding $\epsilon$-stationary points of nonconvex functions with $p$-th order derivative inputs or stochastic gradient inputs, whether with or without the mean-squared smoothness assumption. Technically, we prove our quantum lower bounds by showing that the sequential nature of classical hard instances in all these settings also applies to quantum queries, preventing any quantum speedup other than revealing information of the stationary points sequentially.
\end{abstract}

%%%%%%%%%%%%%%%%%%%%%%%%%%%%%%%%%%%%%%%%%%%%%%%%%%%%%%%%%%%%%%%%%%%%%%%%%%%%%%

\section{Introduction}
Quantum computing is an emerging technology with wide applications. Among those, quantum algorithms for optimization are of general interest. On the one hand, optimization algorithms have wide applications in machine learning, statistics, operations research, and many other areas. On the other hand, it is crucial in quantum computing to figure out the extent of quantum speedups in specific problems, and previous literature had established quantum speedups for solving linear systems~\cite{harrow2009quantum,childs2015quantum}, semidefinite programs~\cite{brandao2016quantum,brandao2017SDP,vanApeldoorn2020quantum,vanApeldoorn2018SDP,kerenidis2020SDP}, general convex optimization~\cite{chakrabarti2020optimization,vanApeldoorn2020optimization}, etc.

More recently, nonconvex optimization has been a primary research direction in machine learning, since the landscapes of many models, including neural networks, are typically nonconvex. Finding the global optimum of a nonconvex function, even approximately, is NP-hard in general~\cite{murty1985some,nemirovski1983problem}. To give efficient optimization algorithms for nonconvex functions, a first step is to find stationary points~\cite{agarwal2017finding,birgin2017worst,carmon2018accelerated,carmon2017convex,nesterov2003introductory,nesterov2006cubic}. However, quantum algorithms for nonconvex optimization are less understood. Based on gradient descents, \citet{zhang2021quantum} proposed a quantum algorithm that can find an $\epsilon$-approximate second-order stationary point of a $d$-dimensional nonconvex function and improve the logarithmic dimension dependence from $\log^6 d$ in the classical result~\cite{jin2018accelerated} to $\log^2 d$, but the $\epsilon$ dependence remains the same. The dependence in $\log d$ is further improved to linear in \citet{childs2022quantum}.
Moreover, \citet{liu2022quantum} showed that quantum tunneling can provide quantum speedups in the task of finding an unknown local minimum starting from a known one.

Meanwhile, various results have been developed concerning the classical lower bounds for finding $\epsilon$-approximate first-order stationary points (i.e., points with gradient smaller than $\epsilon$) of nonconvex functions under different assumptions. In particular, \citet{carmon2020lower} discussed the setting where the objective function $f$ has Lipschitz $p$-th order derivative and proved that any randomized classical algorithm has to make at least $\Omega\big(\epsilon^{-\frac{1+p}{p}}\big)$ derivative queries to guarantee an $\Omega(1)$ success probability in the worst case. Using a similar approach,~\citet{carmon2021lower} proved a deterministic classical lower bound for first-order method of order $\Omega(\epsilon^{-12/7})$ for functions with Lipschitz first and second derivatives.

As for stochastic settings, \citet{arjevani2020second,arjevani2022lower} thoroughly investigated classical lower bounds under different assumptions with or without the mean-squared smoothness property, and with access to different orders of stochastic derivatives. In particular, the query lower bound for stochastic first-order methods is $\Omega(\epsilon^{-4})$~\cite{arjevani2022lower}. If the stochastic gradient additionally satisfies the mean-squared smoothness property, the query lower bound would be of order $\Omega(\epsilon^{-3})$~\cite{arjevani2022lower}, which is also the query lower bound in the case where we have access to second- and higher-order stochastic derivatives~\cite{arjevani2020second}. Nevertheless, for objective functions with Lipschitz $p$-th oder derivative, the query lower bound remains $\Omega(\epsilon^{-3})$ for stochastic $p$-th order methods if we have the mean-squared smoothness property~\cite{arjevani2020second}.

However, despite recent progress on quantum lower bounds for convex optimization~\cite{garg2020no,garg2021near}, quantum lower bounds for nonconvex optimization are widely open.

%====================================================================================

\paragraph{Contributions}
We conduct a systematic study of quantum lower bounds for finding an $\epsilon$-stationary point of a nonconvex objective function $f$, i.e., finding an $\x\in\R^d$ such that
\begin{align}
\|\nabla f(\x)\|\leq\epsilon.\nonumber
\end{align}
For optimization problems with deterministic queries, high-order methods are of general interest~\cite{bubeck2019near,gasnikov2019optimal}, which compared to first-order methods can achieve better convergence rate by exploiting higher-order smoothness~\cite{birgin2017worst,cartis2010complexity,nesterov2006cubic}. Beyond that, another widely considered setting is having access to stochastic gradients, which is widely applied in modern machine learning tasks~\cite{bottou2007tradeoffs,bottou2018optimization} as it only requires access to an unbiased gradient estimator. Various classical algorithms have been developed under this setting from variants of stochastic gradient descent (SGD)~\cite{jin2021nonconvex,fang2018spider,zhang2021escape,zhou2018finding} to more advanced methods~\cite{kingma2014adam,liu2018adaptive,duchi2011adaptive}. Hence, we study the quantum query lower bounds for finding and $\epsilon$-stationary point under the following two important settings:
\begin{enumerate}[leftmargin=*]
\item having access to derivatives of a $p$-th order Lipschitz function;
\item having access to stochastic gradients of a Lipschitz function without the mean-squared smoothness assumption; or
\item having access to stochastic gradients of a Lipschitz function that additionally satisfy the mean-squared smoothness assumption.
\end{enumerate}
For the first setting , we consider a $C^{\infty}$ function $f\colon\R^d\to\R$ with $L_p$-Lipschitz $p$-th derivative, i.e., $\|\nabla^p f(\x)\|\leq L_p$. We define the $p$-th order response to a query at point $\x$ to be
\begin{align}\label{eqn:pth-derivatives}
\nabla^{(0,\ldots,p)}f(\x):=\{f(\x),\nabla f(\x),\ldots,\nabla^p f(\x)\},
\end{align}
which we assume can be accessed via the \emph{quantum evaluation oracle} defined as a unitary map on $\R^d\to\R^{d^0+\cdots+d^p}$ such that for any $\x\in\R^d$,
\begin{align}\label{eqn:Ofp-defn}
O_f^{(p)}\ket{\x}\ket{y}\to\ket{\x}\ket{y\oplus\nabla^{(0,\ldots,p)}f(\x)}.
\end{align}
Here, the Dirac notation $\ket{\cdot}$ denotes the register storing quantum states. Specifically, for any $m\in\mathbb{N}$, $\x_1,\ldots,\x_m\in\R^d$, and $\vect{c}\in\mathbb{C}^m$ such that $\sum_{i=1}^m|c_i|^2=1$,
\begin{align}
O_f^{(p)}\Big(\sum_{i=1}^mc_i\ket{\x_i}\otimes\ket{\0}\Big)=\sum_{i=1}^mc_i\ket{\x_i}\otimes\ket{\nabla^{(0,\ldots,p)}f(\x_i)}.\nonumber
\end{align}
If we measure this quantum state, we get $\nabla^{(0,\ldots,p)}f(\x_i)$ with probability $|c_i|^2$. Compared to the classical evaluation oracle (i.e., $m=1$), the quantum evaluation oracle can query different locations in \textit{superposition}, which is the essence of speedups from quantum algorithms.
In addition, if we can implement the classical evaluation oracle by arithmetic circuits, the quantum evaluation oracle can be implemented by quantum arithmetic circuits of about the same size (see Footnote 2 of~\citealt{chakrabarti2023quantum}). Hence, it has been the standard assumption in previous quantum computing literature for convex optimization~\cite{vanApeldoorn2020optimization,chakrabarti2020optimization} and nonconvex optimization~\cite{liu2022quantum,zhang2021quantum}, with different magnitudes of quantum speedups obtained.

Despite that the quantum evaluation oracle is powerful by taking queries in superposition, however, we show that it cannot provide quantum speedup for finding stationary points of nonconvex functions. In particular, we prove the following quantum query lower bound.
\begin{theorem}[Informal version of \thm{p-th-order-formal}]\label{thm:p-th-order-informal}
For any $\epsilon>0$ and $p\in\mathbb{N}$, there exists a family $\mathcal{F}$ of functions $f\colon\R^d\to\R$ with $L_p$-Lipschitz $p$-th derivative such that any quantum algorithm that finds an $\epsilon$-stationary point of any $f\in\mathcal{F}$ must make 
$$\Omega\left(L_p^{1/p}\epsilon^{-(p+1)/p}\right)$$ queries to the quantum $p$-th order oracle.
\end{theorem}

For the second setting where we have access to stochastic gradients, we also consider a $C^{\infty}$ function $f\colon\R^d\to\R$ that is $L$-Lipschitz, i.e., $\|\nabla f(\x)\|\leq L$. Assume the stochastic gradient $\g(\x,\xi)$ of $f$ indexed by some random seed $\xi$ satisfies
\begin{align}
\begin{cases}
\mathbb{E}_{\xi}[\g(\x,\xi)]=\nabla f(\x), \\
\E_{\xi}\|\g(\x,\xi)-\nabla f(\x)\|^2\leq\sigma^2,\nonumber
\end{cases}
\end{align}
for some constant $\sigma$, which can be accessed via the following quantum oracle
\begin{align}\label{eqn:intro-Og-defn}
O_\g\ket{\x}\ket{\xi}\ket{\vect{v}}=\ket{\x}\ket{\xi}\ket{\g(\x,\xi)+\vect{v}}.
\end{align}
Then, we can prove the following result.

\begin{theorem}[Informal version of \thm{stochastic-formal}]\label{thm:stochastic-informal}
For any $\epsilon,\sigma>0$, there exists a family $\mathcal{F}$ of $L$-gradient Lipschitz functions $f\colon\R^d\to\R$ such that any quantum algorithm that finds an $\epsilon$-stationary point of any $f\in\mathcal{F}$ must make
$$\Omega\left(\frac{\min\{L^2\Delta^2,\sigma^4\}}{\epsilon^4}\right)$$
queries to the quantum stochastic gradient oracle.
\end{theorem}

Moreover, various literature on nonconvex stochastic optimization~\cite{fang2018spider,lei2017non,zhou2018finding} has considered the following additional mean-squared smoothness assumption on the stochastic gradient $\g(\x,\xi)$.
\begin{assumption}\label{assum:mss}
The stochastic gradient $\g$ satisfies that for some constant $\bar{L}$, we have that for all $\x,\y\in\R^d$
\begin{align}
\E_{\xi}\|\g(\x,\xi)-\g(\y,\xi)\|^2\leq\bar{L}^2\cdot\|\x-\y\|^2.\nonumber
\end{align}
\end{assumption}

Under the stochastic optimization setting where \assum{mss} is satisfied, we prove the following result.
\begin{theorem}[Informal version of \thm{stochastic-formal-mss}]\label{thm:stochastic-informal-mss}
For any $\epsilon,\sigma>0$, there exists a family $\mathcal{F}$ of $\bar{L}$-gradient Lipschitz functions $f\colon\R^d\to\R$ such that for any quantum algorithm that finds an $\epsilon$-stationary point of any $f\in\mathcal{F}$ must make $\Omega(\bar{L}\sigma/\epsilon^3)$ queries to the quantum stochastic gradient oracle satisfying the mean-squared smoothness assumption.
\end{theorem}

Observe that our quantum lower bounds match the classical algorithmic results~\cite{birgin2017worst,fang2018spider,jin2021nonconvex} concerning corresponding settings. Therefore, we essentially prove that \textbf{there is no quantum speedup for finding stationary points of nonconvex functions with $p$-th order derivative inputs or stochastic gradient inputs, whether with or without \assum{mss}.}

%=====================================================================================

\paragraph{Techniques}
Inspired by both the classical lower bound results for finding stationary points~\cite{arjevani2022lower,carmon2020lower} as well as the techniques introduced in~\citet{garg2020no,garg2021near} on quantum lower bounds for convex optimization, our work utilizes the underlying similarities and connects these two settings. In particular, the proof of our quantum lower bounds has the following key technical components:

\begin{enumerate}[leftmargin=*]
\item Adopt a hard function that is a robust zero-chain and has a ``non-informative region" around $\0$, which contains a sequential underlying structure that can only be discovered via adaptive queries.
\item Represent quantum algorithms by sequences of unitaries.
\item Demonstrate that the sequential nature of the robust zero-chain can nullify the advantage of quantum algorithms to make queries in \emph{superpositions}.
\end{enumerate}

\textbf{First}, classical hard instances for finding stationary points of nonconvex functions under different settings share the same intuition originated from the following example proposed in Chapter 2.1.2 of \citet{nesterov2003introductory},
\begin{align}\label{eqn:nesterov}
f(\x):=\frac{1}{2}(x_1-1)^2+\frac{1}{2}\sum_{i=1}^{T-1}(x_i-x_{i+1})^2.
\end{align}
For every component $i\in[T]$, $\nabla_i f(\x)=0$ if and only if $x_{i-1}=x_i=x_{i+1}$. Then, if we query a point $\x$ with the first $t$ entries being nonzero, the derivatives $\nabla^{(0,\ldots,p)}f(\x)$ can only reveal the $(t+1)$-th direction. Such $f$ is called a \textit{zero-chain} which is formally defined as follows.

\begin{definition}[\citealt{carmon2020lower}, Definition 3]\label{defn:zero-chain}
For $p\in\mathbb{N}$, a function $f\colon\mathbb{R}^T\to\mathbb{R}$ is called a $p$-th order zero-chain if for every $\x\in\R^d$,
\begin{align}
&\supp\{\x\}\subseteq\{1,\ldots,i-1\}\nonumber\\
&\qquad\Rightarrow\bigcup_{q\in[p]}\supp\{\nabla^qf(\x)\}\subseteq\{1,\ldots,i\}, \nonumber
\end{align}
where the support of a tensor $M\in\mathbb{R}^{\otimes^k T}$ is defined as
\begin{align}
\supp\{M\}:=\{i\in[d]\,|\,M_i\neq 0\}. \nonumber
\end{align}
We say $f$ is a zero-chain if it is a $p$-th-order zero-chain for every $p\in\mathbb{N}$.
\end{definition}

Intuitively, if the objective function with an unknown set of coordinates is a zero-chain, and if we query its derivatives at point $\x$ with only its first $i$ entries being nonzero, such query can only reveal information of the $(i+1)$-th coordinate, exhibiting a sequential nature. Hence, for any classical algorithm that never explores directions with zero derivatives components that seem not to affect the function, which is referred to as ``zero-respecting algorithm"~\cite{carmon2020lower}, it takes at least $T$ queries to learn all the $T$ coordinates. Moreover, we can observe that finding a stationary point of the function $f$ defined in Eq.~\eqn{nesterov} requires complete knowledge of all the $T$ directions, indicating that it takes at least $T$ queries for a zero-respecting algorithm to find the stationary point of $f$.

\citet{carmon2020lower} extended this lower bound to randomized classical algorithms by constructing a hard instance following the intuition of the quadratic hard instance~\eqn{nesterov} and additionally \textbf{creating a ``non-informative" region near $\0$ where small components have no impact on the function value, which can ``trap" random perturbations} since with overwhelming probability they can only create small magnitudes among unknown coordinates and thus has no influence on the function value as well as the algorithm.

\textbf{Second}, \citet{garg2020no,garg2021near} developed quantum lower bounds for convex optimization with non-smooth and smooth objective functions respectively, and demonstrate that there is no quantum speedup in both settings. The hard instance in \citet{garg2020no} is a variant of the shielded Nemirovski function introduced in \citet{bubeck2019complexity} which takes a maximization over several component functions, each related to one of the coordinates and the component function related to the $T$-th coordinate is the least significant. Then with high probability, each query can reveal only one unknown coordinate with the smallest index. This property also applies to the smoothed hard instance in \citet{garg2021near}. To obtain quantum lower bounds, \citet{garg2020no,garg2021near} \textbf{represented quantum algorithms in the form of sequences of unitaries}
\begin{align}
\cdots V_3O_fV_2O_fV_1O_fV_0\nonumber
\end{align}
applied to the initial state $\ket{0}$, where $O_f$ are the evaluation oracle of $f$ and $V_i$s are unitaries that are independent from $f$. The key step in their proof is demonstrating that, for any quantum algorithm $A_{\quan}$ making $k<T$ queries, if we replace all the $k$ queries to $O_f$ by new evaluation oracles that only partly agree with $f$ but contains no information regarding the $T$-th coordinate, the output state of the algorithm will barely change. Since finding an $\epsilon$-stationary point requires knowing all the coordinates, the new sequence of unitaries is hence not be able to find an $\epsilon$-stationary point with high probability, so does the original quantum algorithm $A_{\quan}$.

\textbf{Third}, we observe that although the hard instance in \citet{garg2020no,garg2021near} is a variant of the shielded Nemirovski function and has a different construction compared to the zero-chain hard instance introduced in \citet{carmon2020lower}, it also satisfies the properties of robust zero-chains. We note that this connection has been utilized in previous works~\cite{carmon2020acceleration,woodworth2017lower} but has yet been pointed out explicitly. Similarly, quantum queries to the derivatives of the hard instance of \citet{carmon2020lower} also have only rather limited power, similar to the case in \citet{garg2020no,garg2021near}. Conceptually, this is due to the fact that the hard instance in \citet{carmon2020lower} possesses \textbf{a sequential nature that the $i$-th coordinate direction only emerges when we reach a position that has a large overlap with the $(i-1)$-th direction, which nullifies the unique advantage of quantum algorithms to make queries in superpositions}.

As for the \emph{stochastic setting}, similar to the classical stochastic lower bound result~\cite{arjevani2022lower}, we still use the classical hard instance defined in \citet{carmon2020lower} but with different scaling parameters. Nevertheless, the stochastic gradient function of the hard instance in \citet{arjevani2022lower} could lead to a quantum speedup on this particular hard instance via Grover's search algorithm~\cite{grover1996fast}. To address this issue, inspired by the quantum lower bound on multivariate mean estimation~\cite{cornelissen2022near}, we construct a new stochastic gradient function with details given in \sec{stochastic-construction-to-lowerbound} such that it is also hard for quantum algorithms to obtain an accurate estimation of the exact gradient. Then, a quantum lower bound can be obtained matching the existing classical algorithmic upper bound result following the same procedure as \sec{p-th-order-quantum-lowerbound}.

Furthermore, if we assume that the stochastic gradient function satisfies the \emph{mean-squared smoothness condition described in \assum{mss}}, we can apply a similar version of the function smoothing technique introduced in \citet{arjevani2022lower} to our stochastic gradient function~\eqn{sgf-def} to obtain a ``smoothed" stochastic gradient function, whose detailed formula is given in \sec{mss}, upon which we can obtain a quantum query lower bound matching the existing classical algorithmic upper bound result given that the stochastic gradient function satisfies \assum{mss}.

Recently, a simultaneous work by~\citet{gong2022robustness} proved that finding an $\epsilon$-stationary point of a nonconvex function with a noisy zeroth- and first-order inputs requires $\Omega(\epsilon^{-12/7})$ queries. Technically, they used the hard instance introduced in~\citet{carmon2021lower} that also has a sequential underlying structure such that coordinates can only be revealed sequentially due to the non-informative region near $\0$ created by noise. In contrast, in our setting we do not need external noise to create the non-informative region, which has fundamentally different intuitions. 

%============================================================================================

\paragraph{Open Questions}
Our paper leaves several natural open questions for future investigation:
\begin{itemize}[leftmargin=*]
\item Can we extend our stochastic quantum lower bounds for finding stationary points to higher-order methods?

\item Can we extend our quantum lower bounds to the setting of finding approximate second-order stationary points (i.e., approximate local minima) with no additional overhead, or overhead being at most poly-logarithmic in $\epsilon$ and $d$?

\item In this work, we show that classically hard optimization instances where information can only be revealed sequentially are also hard for quantum algorithms. Can we develop quantum lower bounds for other computational problems with sequential nature via similar techniques?
\end{itemize}

%============================================================================================

%%%%%%%%%%%%%%%%%%%%%%%%%%%%%%%%%%%%%%%%%%%%%%%%%%%%%%%%%%%%%%%%%%%%%%%%%%%%%%

\section{Quantum Lower Bound with Lipschitz $p$-th Order Derivatives}\label{sec:p-th-order-quantum-lowerbound}

\subsection{Function Classes and Classical Lower Bound}\label{sec:p-th-order-classical}
We consider the following set of objective functions.
\begin{definition}[\citealt{carmon2020lower}, Definition 1]\label{defn:Fp}
Let $p\geq 1$, $\Delta>0$ and $L_p>0$. Then the set $\mathcal{F}_p(\Delta, L_p)$
denotes the union, over $d\in\mathbb{N}$, of the collection of $C^{\infty}$ functions $f\colon\mathbb{R}^d\to\mathbb{R}$ with $L_p$-Lipschitz $p$-th derivative and $f(\0)-\inf_\x f(\x)\leq\Delta$.
\end{definition}
For any $f\in\mathcal{F}_p(\Delta,L_p)$, the response of a $p$-th order oracle to a query at point $\x$ is
\begin{align*}
\nabla^{(0,\ldots,p)}f(\x)=\{f(\x),\nabla f(\x),\ldots,\nabla^p f(\x)\}
\end{align*}
as defined in Eq.~\eqn{pth-derivatives}. Then for any dimension $d\in\mathbb{N}$, a classical algorithm $A$ is defined as a map from objective functions $f\in\mathcal{F}_p(\Delta, L_p)$ to a sequence of iterates in $\mathbb{R}^d$, if for any $i\in\mathbb{N}$ it produces iterates of the form
\begin{align*}
\x^{(i)}=A^{(i)}\big(\nabla^{(0,\ldots,p)}f(\x^{(1)}),\ldots,\nabla^{(0,\ldots,p)}f(\x^{(i-1)})\big)
\end{align*}
where $A^{(i)}$ is a measurable mapping to $\mathbb{R}^d$. We refer to $A$ as a classical $p$-th-order deterministic algorithm.

Similarly, a classical $p$-th-order randomized algorithm $A_{\text{rand}}^{(p)}$ as a distribution on $p$-th order deterministic algorithms. Quantitatively, $A_{\text{rand}}^{(p)}$ would produce iterates of the form
\begin{align*}
\x^{(i)}=A^{(i)}\big(\xi,\nabla^{(0,\ldots,p)}f(\x^{(1)}),\ldots,\nabla^{(0,\ldots,p)}f(\x^{(i-1)})\big)
\end{align*}
for any $i\in\mathbb{N}$, where $\xi$ is a random uniform variable on $[0,1]$, for some measurable mappings $A^{(i)}$ into $\mathbb{R}^d$.

%=====================================================================================================

The strategy of constructing a classical hard instance is to construct a high-dimensional zero-chain (\defn{zero-chain}) such that the position of its stationary point is related to all the coordinates. In particular, \citet{carmon2020lower} considered the following hard instance $\bar{f}_T(\x)\colon\mathbb{R}^d\to\mathbb{R}$,
\begin{align}\label{eqn:bar-fT}
\bar{f}_T(\x)&=-\Psi(1)\Phi(x_1)\nonumber\\
&+\sum_{i=2}^T[\Psi(-x_{i-1})\Phi(-x_i)-\Psi(x_{i-1})\Phi(x_i)],
\end{align}
where
\begin{align}
\Psi(x)&:=
\begin{cases}
0 \qquad\qquad\qquad\quad\ \ \, x\leq 1/2, \nonumber\\
\exp\big(1-\frac{1}{(2x-1)^2}\big)\quad x>1/2,
\end{cases}
\\
\Phi(x)&:=\sqrt{e}\int_{-\infty}^x e^{-t^2/2}\d t,
\nonumber
\end{align}
for some $T>0$. Note that $\bar{f}_T$ is a zero-chain whose derivative $\nabla\bar{f}_T(\x)$ has a large norm unless $|x_i|\geq 1$ for all $i\in[T]$ (see \lem{fT-large-gradient} for details). Thus, it takes at least $\Omega(T)$ queries for a classical algorithm to find a stationary point of $\bar{f}_T$ if it never explores directions with zero derivatives, which is referred to as a \textit{zero-respecting algorithm} in~\citet{carmon2020lower}, where the authors also showed that the query complexity lower bound for zero-respecting algorithms also holds for all deterministic classical algorithms, and $T$ can at most be of order $O\left(\Delta L_p^p\epsilon^{-(1+p)/p}\right)$ to satisfy $L_p$-Lipschitzness and other conditions, establishing query lower bound for all deterministic classical algorithms.

%================================================================================================================================

This deterministic lower bound is further extended to obtain a distributional complexity lower bounds of all randomized classical algorithms by a simple random orthogonal transformation (or intuitively, a high-dimensional random rotation) on the hard instance $\bar{f}_T$ defined in Eq.~\eqn{bar-fT}:
\begin{align}\label{eqn:tildef-defn}
\tildef(\x)\coloneqq\alpha\bar{f}_T(U^T\x/\beta)
\end{align}
where $\alpha$ and $\beta$ are scaling constants, $U\in\mathbb{R}^{d\times T}$ with columns $\u^{(1)},\ldots,\u^{(T)}$ is an orthogonal matrix with $T\leq d$, and we assume throughout that $U$ is chosen uniformly at random from the space of orthogonal matrices $O(d,T)=\{U\in\mathbb{R}^{d\times T}\,|\,U^TU=I_T\}$.

It is shown in \citet{carmon2020lower} that any random algorithm can ``discover" at most one coordinate $\u^{(i)}$ per query with high probability. Quantitatively, for a random orthogonal matrix $U$, any sequence of bounded iterates $\{\x^{(t)}\}_{t\in\mathcal{N}}$ based on derivatives of $\tildef$ must satisfy $|\<\x^{(t)},\u^{(j)}\>|\leq0.5$ with high probability for all $t$ and $j>t$. Then by \lem{fT-large-gradient} in \append{existing-lemmas}, with high probability $\|\nabla\tildef(\x^{(t)})\|$ is large for any $t\leq T$, and thus establishes the lower bound on randomized algorithms with access to bounded iterates. Moreover, \citet{carmon2020lower} showed that the boundedness of the iterates can be removed by composing $\tildef$ with a soft projection to reach a lower bound for general unbounded iterates, and the query lower bound $\Omega\big(\Delta L_p^p\epsilon^{-(1+p)/p}\big)$ can be obtained for all randomized classical algorithms with access to $p$-th-order derivatives.

%===========================================================================================================================

\subsection{Quantum Query Model and Complexity Measures}\label{sec:quantum-model}
We adopt the quantum query model introduced in \citet{garg2020no}. For a function $f\colon\R^d\to\R$ with $L_p$-Lipschitz $p$-th derivative, we assume access to the following quantum oracle $O_f^{(p)}$ defined in Eq.\eqn{Ofp-defn}:
\begin{align}
O_f^{(p)}\ket{\x}\ket{y}\to\ket{\x}\ket{y\oplus\nabla^{(0,\ldots,p)}f(\x)},\nonumber
\end{align}
for $\nabla^{(0,\ldots,p)}f(\x)$ defined in Eq.~\eqn{pth-derivatives}. Then for any $p$-th order quantum query algorithm $A_{\quan}$, it can be described by the following sequence of unitaries
\begin{align}\label{eqn:quantum-algorithm-form}
\cdots V_3O_f^{(p)}V_2O_f^{(p)}V_1O_f^{(p)}V_0
\end{align}
applied to an initial state, which can be set to $\ket{0}$ without loss of generality. Moreover, we define $A_{\quan}^{(t)}$ to be the sequence \eqn{quantum-algorithm-form} truncated before the $(t+1)$-th query to $O_f^{(p)}$,
\begin{align}
A_{\quan}^{(t)}:=V_{t}O_f^{(p)}\cdots O_f^{(p)}V_1O_f^{(p)}V_0,\nonumber
\end{align}
for any $t\in\N$. Next, we extend the classical complexity measure introduced in \citet{carmon2020lower} to the quantum regime. Quantitatively, 
\begin{definition}[Quantum complexity measures]\label{defn:quantum-complexity-measure}
For a function $f\colon\R^d\to\R$ and a sequence of quantum states $\big\{\ket{\psi}^{(t)}\big\}_{t\in\N}$, let $p_t$ be the probability distribution over $\x\in\R^d$ obtained by measuring the state $\ket{\psi}^{(t)}$ in the computational basis $\{\ket{\x}|\,\x\in\R^d\}$. Then we can define
\begin{align}
&T_{\epsilon}\big(\big\{\ket{\psi}^{(t)}\big\}_{t\in\N},f\big):=\nonumber\\
&\qquad\inf\Big\{t\in\mathbb{N}\,|\,\Pr_{\x\sim p_t}\big(\|\nabla f(\x)\|\leq\epsilon\big)\geq\frac{1}{3}\Big\}.\nonumber
\end{align}
\end{definition}
\vspace{-2mm}

To measure the performance of a quantum algorithm $A_{\quan}$ on function $f$, we define
\begin{align}
T_{\epsilon}(A_{\quan},f):=T_{\epsilon}\big(\big\{A_{\quan}^{(t)}\ket{0}\big\},f\big)\nonumber
\end{align}
as the complexity of $A_{\quan}$ on $f$. With this setup, we define the complexity of algorithm class $\mathcal{A}_{\quan}$ of all quantum algorithms in the form \eqn{quantum-algorithm-form} on a function class $\mathcal{F}$ to be
\begin{align}\label{eqn:complexity-measure}
\mathcal{T}_{\epsilon}(\mathcal{A}_{\quan},\mathcal{F}):=\inf_{A\in\mathcal{A}_{\quan}}\sup_{f\in\mathcal{F}}T_{\epsilon}(A,f).
\end{align}

%===========================================================================================================================

\subsection{Quantum Lower Bound with Bounded Input Domain}\label{sec:noiseless-lowerbound}
We first prove a query complexity lower bound for any quantum algorithm $A_{\quan}$ defined in \sec{quantum-model} on a function class with bounded input domain using the hard instance $\tildef$ defined in Eq.~\eqn{tildef-defn}. For the convenience of notations, we use $\widetilde{O}^{p}_{T;U}$ to denote the quantum evaluation oracle encoding the $p$-th-order derivatives of function $\tildef$, or equivalently
\vspace{-2mm}
\begin{align}
\widetilde{O}^{(p)}_{T;U}\ket{\x}\ket{y}\to\ket{\x}\ket{y\oplus\nabla^{(0,\ldots,p)}\tildef(\x)}.\nonumber
\end{align}
Consider the truncated sequence $A_{\quan}^{(k)}$ of any possible quantum algorithm $A_{\quan}$ with $k<T$, we define a sequence of unitaries starting with $A_0=A_{\quan}^{(k)}$ as follows:
\begin{align}
A_0&:=V_k\tildeOp_{T;U}V_{k-1}\tildeOp_{T;U}\cdots \tildeOp_{T;U}V_1\tildeOp_{T;U}V_0\nonumber \\
A_1&:=V_k\tildeOp_{T;U}V_{k-1}\tildeOp_{T;U}\cdots \tildeOp_{T;U}V_1\tildeOp_{1;U_1}V_0\nonumber \\
A_2&:=V_k\tildeOp_{T;U}V_{k-1}\tildeOp_{T;U}\cdots \tildeOp_{2;U_2}V_1\tildeOp_{1;U_1}V_0 \label{eqn:unitary-sequences} \\
&\vdots \nonumber\\
A_k&:=V_k\tildeOp_{k;U_k}V_{k-1}\tildeOp_{k-1;U_{k-1}}\cdots \tildeOp_{2;U_2}V_1\tildeOp_{1;U_1}V_0, \nonumber
\end{align}
where $\tildeOp_{t,U_t}$ stands for the evaluation oracle of function $\tilde{f}_{t;U_t}$ and its $p$-th-order derivatives defined as
\begin{align}\label{eqn:tildeft-defn}
\tilde{f}_{t;U_t}(\x):=\alpha\bar{f}_t(U_t^T\x/\beta)
\end{align}
Our goal to show that the algorithm $A_0$ does not solve our problem. To achieve that, we follow a similar approach shown in~\citet{garg2020no} and develop a hybrid argument in which we first show that the outputs of the algorithm $A_i$ and $A_{i+1}$ are close, so does the outputs of $A_0$ and $A_k$. Then, we argue that the algorithm $A_k$ cannot find an $\epsilon$-stationary point with high probability since oracles in the algorithm are independent from $\u_T$. Hence, $A_k$ cannot do better than random guessing a vector $\u_T$, which by \lem{cannot-guess} in \append{probabilistic-facts} fails with overwhelming probability. 
\begin{lemma}[$A_t$ and $A_{t-1}$ have similar outputs]\label{lem:similar-outputs}
For a hard instance $\tildef(\x)\colon\R^d\to\R$ defined on $\mathbb{B}(\0,2\beta\sqrt{T})$ with $d\geq 200T\log T$, let $A_t$ for $t\in[k-1]$ be the unitaries defined in \eqn{unitary-sequences}. Then
\begin{align}
\E_U\big(\|A_t\ket{\0}-A_{t-1}\ket{\0}\|^2\big)\leq 1/36T^{4}.\nonumber
\end{align}
\end{lemma}
\vspace{-5mm}

\begin{proof}
Since the series of unitaries in Eq.~\eqn{unitary-sequences} was constructed by gradually changing the quantum evaluation oracle, the difference between consecutive terms can be expressed as
\vspace{-2mm}
\begin{align}
&\|A_t\ket{\0}-A_{t-1}\ket{\0}\|\nonumber\\ &\quad=\big\|\big(\tildeOp_{t;U_t}-\tildeOp_{T;U_T}\big)V_{t-1}\tildeOp_{t-1;U_{t-1}}\cdots\tildeOp_{1;U_1}V_0\ket{\0}\big\|\nonumber.
\end{align}
We will prove the claim for any fixed choice of vectors $\{\u^{(1)},\ldots,\u^{(t-1)}\}$, which will imply the claim for any distribution over those vectors. After fixing these vectors, we can see that the quantum state
\begin{align}
V_{t-1}\tildeOp_{t-1;U_{t-1}}\cdots\tildeOp_{1;U_1}V_0\ket{\0} \nonumber
\end{align}
is fixed and we refer to it as $\ket{\psi}$. Thus our problem reduces to showing for all quantum states $\ket{\psi}$,
\begin{align}
&\E_{\{\u^{(t)},\ldots,\u^{(T)}\}}\big(\|\big(\tildeOp_{t;U_t}-\tildeOp_{T;U_T}\big)\ket{\psi}\|^2\big)\leq \frac{1}{36T^4}. \nonumber
\end{align}
For any $\ket{\psi}$, it can be expressed as $\ket{\psi}=\sum_\x\alpha_\x\ket{\x}\ket{\phi_\x}$, where $\x$ is the query made to the oracle, and $\sum_\x|\alpha_\x|^2=1$, which leads to
\vspace{-2mm}
\begin{align}
&\E_{\{\u^{(t)},\ldots,\u^{(T)}\}}\Big(\Big\|\sum_\x\alpha_\x \big(\tildeOp_{t;U_t}-\tildeOp_{T;U_T}\big)\ket{\x}\ket{\phi_\x}\Big\|^2\Big)\nonumber\\
&\qquad\leq\sum_\x|\alpha_\x|^2\E_{\{\u^{(t)},\ldots,\u^{(T)}\}}\big(\big\|\big(\tildeOp_{t;U_t}-\tildeOp_{T;U_T}\big)\ket{\x}\big\|^2\big).\nonumber
\end{align}
Since $|\alpha_\x|^2$ defines a probability distribution over $\x$, we can again upper bound the right hand side for any $\x\in\mathbb{B}(\0,2\beta\sqrt{T})$ instead. Since $\tildeOp_{t;U_t}$ and $\tildeOp_{T;U_T}$ behave identically for some inputs x, the only nonzero
terms are those where the oracles respond differently, which can only happen if
\begin{align}
\nabla^{(0,\ldots,p)}\tildef(\x)\neq\nabla^{(0,\ldots,p)}\tilde{f}_{t;U_t}(\x). \nonumber
\end{align}
When the response is different, we can upper bound $\big\|\big(\tildeOp_{t;U_t}-\tildeOp_{T;U}\big)\ket{\x}\|^2$ by 4 using the triangle inequality. Thus for any $\x\in\mathbb{B}(\0,\sqrt{T})$, we have
\begin{align}
&\E_{\{\u^{(t)},\ldots,\u^{(T)}\}}\big(\big\|\big(\tildeOp_{t;U_t}-\tildeOp_{T;U_T}\big)\ket{\x}\big\|^2\big)\nonumber\\
&\quad\leq 4\Pr_{\{\u^{(t)},\ldots,\u^{(T)}\}}\big(\nabla^{(0,\ldots,p)}\tildef(\x)\neq\nabla^{(0,\ldots,p)}\tilde{f}_{t;U_t}(\x)\big)\nonumber\\
&\quad\leq \frac{1}{36T^4}, \nonumber
\end{align}
where the last inequality follows from \lem{quantum-zero-chain}.
\end{proof}
Based on \lem{similar-outputs}, we can obtain the following result with its proof deferred to \append{thm-p-th-order-formal}.
%\tyl{broken link here}\cyz{fixed}.
By combining \lem{similar-outputs} and the soft projection technique in \citet{carmon2020lower} that can remove the boundedness of the iterates, we obtain the following quantum lower bound.

\begin{theorem}[Formal version of \thm{p-th-order-informal}]\label{thm:p-th-order-formal}
There exist numerical constants $0<c_0,c_1<\infty$ such that the following lower bound holds. Let $p\geq 1$, $p\in \mathbb{N}$, and let $\Delta$, $L_p$, and $\epsilon$ be positive. Then,
\vspace{-2mm}
\begin{align}
\mathcal{T}_{\epsilon}\big(\mathcal{A}_{\quan},\mathcal{F}_p(\Delta,L_p)\big)\geq c_0\Delta\left(\frac{L_p}{\ell_p}\right)^{1/p}\epsilon^{-\frac{1+p}{p}},\nonumber
\end{align}
where $\ell_p\leq e^{\frac{5}{2}p\log p+c_1p}$, the complexity measure $\mathcal{T}_{\epsilon}$ is defined in Eq.~\eqn{complexity-measure}, and the function class $\mathcal{F}_p(\Delta,L_p)$ is defined in \defn{Fp}. The lower bound holds even if we restrict $\mathcal{F}_p(\Delta,L_p)$ to functions whose domain has dimension
\vspace{-2mm}
\begin{align}
\Omega\left(\frac{200c_0\Delta L_p^{1/p}}{\ell_p^{1/p}}\epsilon^{-\frac{1+p}{p}}\cdot\log\Big(\frac{c_0\Delta L_p^{1/p}}{\ell_p^{1/p}}\epsilon^{-\frac{1+p}{p}}\Big)\right).\nonumber
\end{align}
\end{theorem}

%%%%%%%%%%%%%%%%%%%%%%%%%%%%%%%%%%%%%%%%%%%%%%%%%%%%%%%%%%%%%%%%%%%%%%%%%%%%%%%%%%%%%%%%%%%%%%%%%%%%%%%

\section{Quantum Lower Bounds with Access to Stochastic Gradients}\label{sec:stochastic-quantum-lowerbound}
Based on~\citet{carmon2020lower}, \citet{arjevani2020second,arjevani2022lower} further investigated the classical query lower bound finding $\epsilon$-stationary points given access to stochastic first-order or additionally higher-order derivatives.

%==================================================================

\subsection{Classical Hard Instance and Lower Bound}\label{sec:stochastic-classical}
Adopt the notation in \citet{arjevani2022lower}, in this subsection we discuss lower bounds for quantum algorithms finding stationary points of functions in the set $\mathcal{F}(\Delta,L)$ such that for any $F\colon\R^d\to\R$ with $F\in\mathcal{F}$ we have
\vspace{-2mm}
\begin{align}
F(\0)-\inf_\x F(\x)\leq\Delta,\nonumber
\end{align}
and
\vspace{-2mm}
\begin{align}
\|\nabla F(\x)-\nabla F(\vect{y})\|\leq L\|\x-\vect{y}\|,\quad\forall\x,\y\in\R^d.\nonumber
\end{align}

\citet{arjevani2022lower} proved the classical lower bound by extending the zero-chain property introduced in \citet{carmon2020lower} to the stochastic setting.

\begin{definition}[Probability-$p$ zero chain]
A stochastic gradient function $\g(\x,\xi)$ is a probability-$p$ zero-chain if
\begin{align}
&\Pr\big(\exists\x:\prog_0(\g(\x,\xi))=\prog_{\frac{1}{4}}(\x)+1\big)\leq p \nonumber\\
&\Pr\big(\exists\x:\prog_0(\g(\x,\xi))>\prog_{\frac{1}{4}}(\x)+1\big)=0\nonumber
\end{align}
are both satisfied, where
\begin{align}\label{eqn:defn-prog}
\prog_\zeta(\x):=\max\{i\geq 0\,|\,|x_i|\geq\zeta\},
\end{align}
with $x_0\equiv 0$ for notation consistency
\end{definition}

Intuitively, for a probability-$p$ zero-chain with dimension $T$, any zero-respecting stochastic algorithm takes $\Omega(1/p)$ queries to discover one new coordinate in expectation. Hence, the expected number of queries to discover all the coordinates is $\Omega(T/p)$.

\citet{arjevani2022lower} also used the same underlying function $\bar{f}_T(\x)$ defined in Eq.~\eqn{bar-fT} with the following stochastic gradient function %\tyl{equation too long, fix}\cyz{fixed}
\begin{align}\label{eqn:classical-SG}
&[\g_T(\x,\xi)]_i:=\nonumber\\
&\qquad\nabla_i\bar{f}_T(\x)\left(1+\mathbbm{1}\{i>\prog_{\frac{1}{4}}(\x)\}\Big(\frac{\xi}{p}-1\Big)\right),
\end{align}
where $\xi\sim\text{Bernoulli}(p)$ with $p=O(\epsilon^2)$ in \citet{arjevani2022lower}. Then, $\bar{f}_T$ together with $\g_T$ forms a probabilistic-$\frac{1}{4}$ zero chain. With a similar approach in \sec{p-th-order-quantum-lowerbound}, a lower bound of order $\Omega(1/\epsilon^4)$ for all stochastic first-order algorithms can thus be obtained.

In this work we show that, although the classical hard instance in \citet{arjevani2022lower} can be solved via $\tilde{O}(1/\epsilon^3)$ quantum stochastic gradient queries,\footnote{Throughout this paper, the $
\tilde{O}$ and $\tilde{\Omega}$ notations omit poly-logarithmic terms, i.e., $\tilde{O}(g)=O(g\poly(\log g))$ and $\tilde{\Omega}(g)=\Omega(g\poly(\log g))$.} there exist harder instances such that any quantum algorithm also needs to make at least $\Omega(1/\epsilon^4)$ to guarantee a high success probability in the worst case. If the stochastic gradients additionally satisfy the mean-squared smoothness condition in \assum{mss}, we can show that any quantum algorithm needs to make at least $\Omega(1/\epsilon^3)$ to find an $\epsilon$-approximate stationary point with high probability in the worst case. These two quantum lower bounds concerning stochastic gradients satisfying \assum{mss} or not match the corresponding classical algorithmic upper bounds and are thus tight.

%=========================================================================

\subsection{Stochastic Quantum Query Model and Quantum Speedup on the Classical Hard Instance}\label{sec:stochastic-quantum-model}

We can extend our quantum query model introduced in \sec{quantum-model} to the stochastic settings. For a $d$-dimensional, $L$-smooth objective function $f$, we assume access to the quantum stochastic gradient oracle $O_f^{(p)}$ defined as follows:
\begin{definition}[Quantum stochastic gradient oracle]\label{defn:quantum-SG-oracle}
For any $L$-lipschitz function $f\colon\R^d\to\R$, its quantum stochastic first-order oracle $O_\g$ consists of a distribution $P_{\xi}$ and an unbiased mapping $O_\g$ satisfying
\begin{align}
O_\g\ket{\x}\ket{\xi}\ket{\vect{v}}=\ket{\x}\ket{\xi}\ket{\g(\x,\xi)+\vect{v}},\nonumber
\end{align}
where the stochastic gradient $\g(\x,\xi)$ satisfies both
\begin{align}
&\mathbb{E}_{\xi\sim P_{\xi}}[\g(\x,\xi)]=\nabla f(\x) \nonumber\\
&\mathbb{E}_{\xi\sim P_{\xi}}[\|\g(\x,\xi)-\nabla f(\x)\|^2]\leq\sigma^2 \label{eqn:bounded-variance}
\end{align}
for some constant $\sigma$.
\end{definition}
\vspace{-2mm}

To measure the performance of a quantum algorithm $A_{\quan}$ on function $f$ with queries to its stochastic gradient oracle $O_{\g}$, based on \defn{quantum-complexity-measure} we define
\begin{align}
T_{\epsilon}(A_{\quan},f):=T_{\epsilon}\big(\big\{A_{\quan}^{(t)}\ket{0}\big\},f\big)\nonumber
\end{align}
as the complexity of $A_{\quan}$ on $f$. With this setup, we define the complexity of algorithm class $\mathcal{A}_{\quan}$ of all quantum algorithms in the form \eqn{quantum-algorithm-form} on a function class $\mathcal{F}$ to be
\vspace{-2mm}
\begin{align}\label{eqn:stochastic-complexity-measure}
\mathcal{T}^{\sto}_{\epsilon}(\mathcal{A}_{\quan},\mathcal{F},\sigma):=\inf_{A\in\mathcal{A}_{\quan}}\sup_{f\in\mathcal{F}}\sup_{\g}T_{\epsilon}(A,f),
\end{align}
where the last supremum is over all possible stochastic gradient functions $\g$ of $f$ satisfying the bounded-variance requirement \eqn{bounded-variance} in \defn{quantum-SG-oracle}.

Based on this complexity measure, we show that quantum algorithm can find an $\epsilon$-approximate stationary point of the classical hard instance based on $\bar{f}_T$ with fewer queries than the classical lower bound of order $\Omega(1/\epsilon^4)$ by approximating the exact gradient via Grover's algorithm. In particular, we prove the following result.

\begin{lemma}\label{lem:quantum-speedup-on-classical-hard}
Consider the classical hard instance in~\citet{arjevani2022lower} obtained by projecting $\bar{f}_T$ defined in Eq.~\eqn{bar-fT} into a $d$-dimensional space while adopting the stochastic gradient function $\g_T(\x,\xi)$ defined in Eq.~\eqn{classical-SG}, there exists a quantum algorithm using $\tilde{O}(1/\epsilon^3)$ queries to the stochastic quantum gradient oracle $O_{\g}$ defined in \defn{quantum-SG-oracle} that can find find a point $\x$ satisfying $\bar{f}_T(\x)=0$ with probability at least $1/2$.
\end{lemma}
\vspace{-2mm}

The proof of \lem{quantum-speedup-on-classical-hard} is deferred to \append{special-case-quantum-speedup}.

%=========================================================================

\subsection{Quantum Lower Bound}\label{sec:stochastic-construction-to-lowerbound}

In this subsection, we introduce a new hard instance where any quantum algorithm also have to make $\Omega\big(\epsilon^{-4}\big)$ queries to the quantum stochastic oracle defined in \defn{quantum-SG-oracle} to find an $\epsilon$-stationary point with high probability.

%========================================================================================================================

Before presenting the construction of the stochastic gradient function, we first review some existing results on quantum multivariate mean estimation.
\begin{problem}\label{prob:multivariate-mean}
Consider a matrix $M\in\R^{d\times 2T}$ for some $T\leq d/4$, denote $\vect{m}^{(j)}$ as the $i$-th column of $M$. Suppose $T$ columns of $M$ forms a set of orthonormal vectors, while the other $T$ columns are all zero. The goal is to get a good estimation of the direction of the vector $
\g:=\frac{1}{2T}\sum_{j=1}^{2T}\vect{m}^{(j)}$. Formally, we define
\begin{align}
W_\g(\eta)=\bigg\{\ket{\tilde{\g}}\,\Big|\,\frac{|\<\tilde{\g},\g\>|}{\|\tilde{\g}\|\cdot\|\g\|}\geq 1-\eta,\quad\tilde{\g}\in\R^d\bigg\},\nonumber
\end{align}
and define $\Pi_\g(\eta)$ to be the projection operator onto $W_\g(\eta)$. For some small $\eta$, the goal is to produce a quantum state $\ket{\psi}$ with large value of $\|\Pi_\g(\eta)\cdot\ket{\psi}\|$, given the following quantum oracle
\begin{align}\label{eqn:O_m}
O_{M}\ket{\x}\ket{j}\ket{\vect{v}}=\ket{\x}\ket{j}\ket{\vect{m}^{(j)}+\vect{v}}.
\end{align}
\end{problem}

\begin{lemma}[\citealt{cornelissen2022near}, Theorem 3.7]\label{lem:multivariate-mean-estimation}
For any $n<T/2$ and any quantum algorithm that uses at most $n$ queries to the quantum oracle $O_{M}$ defined in \eqn{O_m} with output state $\ket{\psi}$, we have
\begin{align}
\mathbb{E}_{M}[\|\Pi_{\g}(\eta)\cdot\ket{\psi}\|^2]\leq \exp(-\zeta T),\nonumber
\end{align}
for some small constant $\zeta$, and
\begin{align}\label{eqn:mean-estimator-eta}
\eta\geq\frac{3\sqrt{2}}{8\|\g\|}\sqrt{\frac{\E_j\big[\|\m^{(j)}-\g\|^2\big]}{T}}\geq\frac{3}{4}.
\end{align}
and the expectation is over all possible matrices $M$.
\end{lemma}
\vspace{-2mm}

Drawing inspiration from the hard instance for quantum multivariate mean estimation introduced by \citet{cornelissen2022near}, we design our stochastic gradient function accordingly. In a manner similar to the approach employed by \citet{arjevani2022lower}, we incorporate stochasticity to amplify the difficulty of achieving significant progress in individual coordinates through stochastic gradient information. Specifically, for the $d$-dimensional function $\tilde{f}{T;U}$, where $d\geq 4T$, we ensure that for any point $\x$ with gradient $\g(\x)$, there exists a matrix $M_{\x}\in\R^{d\times 2T}$. This matrix possesses $T$ columns consisting of zeros ($\0$), while the remaining $T$ columns form a set of orthonormal vectors that satisfy
\vspace{-2mm}
\begin{align}\label{eqn:M-construction}
\nabla_{\prog_{\frac{\beta}{4}}(\x)+1}\tildef(\x)=\frac{1}{2T}\sum_j 2\gamma\sqrt{T}\cdot\vect{m}_{\x}^{(j)},
\end{align}
where $\vect{m}_{\x}^{(j)}$ stands for the $j$-th column of $M_\x$ and
\vspace{-2mm}
\begin{align}
\gamma=\big\|\nabla_{\prog_{\frac{\beta}{4}}(\x)+1}\tildef(\x)\big\|\leq 23\nonumber
\end{align}
is the norm of the $(\prog_{\beta/4}(\x)+1)$-th gradient component at certain points whose exact value is specified later.

Moreover, to guarantee that all the stochastic gradients at $\x$ can only reveal the $(\prog_{\beta/4}(\x)+1)$-th coordinate $\u_{\prog_{\beta/4}(\x)+1}$ even with infinite number of queries and will not ``accidentally" make further progress, we additionally require that for any $\x,\y\in\R^d$ with $\prog_{\beta/4}(\x)\neq\prog_{\beta/4}(\y)$, all the columns of $M_\x$ are orthogonal to all the columns of $M_\y$. This can be achieved by creating $T$ orthogonal subspaces
\vspace{-3mm}
\begin{align}
\{\mathcal{V}_1,\ldots,\mathcal{V}_T\},\nonumber
\end{align}
where each subspace is of dimension $2T$ and has no overlap with $\{\u_1,\ldots,\u_T\}$, such that the columns of $M_\x$ are within
\vspace{-3mm}
\begin{align}
\spn\big\{\u_{\prog_{\frac{\beta}{4}}(\x)+1},\mathcal{V}_{\prog_{\frac{\beta}{4}}(\x)+1}\big\},\nonumber
\end{align}
as long as the dimension $d$ is larger than $2T^2+T=O(T^2)$.

Then, we can define the following stochastic gradient function for $\nabla\tildef(\x)$:
\vspace{-3mm}
\begin{align}\label{eqn:quantum-hard-g}
\hspace{-4mm}\g(\x,j)=\g(\x)-\g_{\prog_{\beta/4}(\x)+1}(\x)+2\gamma\sqrt{T}\cdot\vect{m}_{\x}^{(j)}
\end{align}
where $j$ is uniformly distributed in the set $[2T]$. Then based on \lem{multivariate-mean-estimation}, we know that it is hard for quantum algorithms to get a accurate estimation of the direction of $\g_{\prog_{\beta/4}(\x)+1}(\x)$ given only access to stochastic gradients at point $\x$ defined in Eq.~\eqn{quantum-hard-g} using less than $T/2$ queries.

Further, we can show that if one only knows about the first $t$ components $\{\u^{(1)},\ldots,\u^{(t)}\}$, even if we permit the quantum algorithm to query the stochastic gradient oracle at different positions of $\x$, it is still hard to learn $\u^{(t+1)}$ as well as other components with larger indices. Quantitatively, for any $1\leq t\leq T$ we define
\begin{align}\hspace{-3mm}\label{eqn:W_t_perp-defn}
W_{t;\perp}:=\Big\{\x\in&\mathbb{B}(\0,\beta\sqrt{T})\,\big|\,\exists i,\text{ s.t. }\nonumber\\
&|\<\x,\u^{(q)}\>|\geq\frac{\beta}{4}\text{ and }t<i\leq T\Big\},
\end{align}
and
\begin{align}
W_{i;\parallel}:=\mathbb{B}(\0,\beta\sqrt{T})-W_{i;\perp}.\nonumber
\end{align}
Intuitively, $W_{t;\perp}$ is the subspace of $\mathbb{B}(\0,\beta\sqrt{T})$ such that any vector in $W_{t;\perp}$ has a relatively large overlap with at least one of $\u^{(t+1)},\ldots,\u^{(T)}$. Moreover, we use $\Pi_{t;\perp}$ and $\Pi_{t;\parallel}$ to denote the quantum projection operators onto $W_{t;\perp}$ and $W_{t;\parallel}$, respectively. As shown in \lem{gradient-estimation-no-speedup}, if starting in the subspace $W_{t;\parallel}$, any quantum algorithm using at most $T/2$ queries at arbitrary locations cannot output a quantum state that has a large overlap with $W_{t;\perp}$ in expectation. Then, we can obtain the following result.

\begin{theorem}[Formal version of \thm{stochastic-informal}]\label{thm:stochastic-formal}
For any $\Delta$, $L$, $\sigma$, and $\epsilon$ that are all positive, we have
\begin{align}
\hspace{-3mm}\mathcal{T}^{\sto}_{\epsilon}\big(\mathcal{A}_{\quan},\mathcal{F}_1(\Delta,L),\sigma\big)
\geq \Omega\Big(\frac{\min\{L^2\Delta^2,\sigma^4\}}{\epsilon^4}\Big),\nonumber
\end{align}
where the complexity measure $\mathcal{T}^{\sto}_{\epsilon}(\cdot)$ is defined in Eq.~\eqn{stochastic-complexity-measure}, and the function class $\mathcal{F}_1(\Delta,L)$ is defined in \defn{Fp}. The lower bound still holds even if we restrict $\mathcal{F}_1(\Delta,L)$ to functions whose domain has dimension $
O(\min\{L^2\Delta^2,\sigma^4\}/\epsilon^4)$.
\end{theorem}
\begin{remark}
Compared to the classical result~\cite{arjevani2022lower}, the lowest possible dimension of the hard instance is improved from $\Theta(\epsilon^{-6})$ to $\Theta(\epsilon^{-4})$, which is due to a sharper analysis and may be of independent interest.
\end{remark}
The proof of \thm{stochastic-formal} is deferred to \append{stochastic-lowerbound-unbounded}.

%====================================================================================================================

\subsection{Quantum Lower Bound with the Mean-Squared Smoothness Assumption}\label{sec:mss}

We also prove a quantum query lower bound for finding an $\epsilon$-stationary point with access to the quantum stochastic gradient oracle defined in \defn{quantum-SG-oracle} and additionally satisfies the \textit{mean-squared smoothness} assumption defined in \assum{mss} for some constant $\bar{L}$.

\begin{theorem}[Formal version of \thm{stochastic-informal-mss}]\label{thm:stochastic-formal-mss}
For any $\Delta$, $\bar{L}$, $\sigma$, and $\epsilon$ that are all positive, we have
\begin{align}
\mathcal{T}^{\sto}_{\epsilon}\big(\mathcal{A}_{\quan},\mathcal{F}_1(\Delta,\bar{L}),\sigma\big)
\geq \Omega(\Delta\bar{L}\sigma/\epsilon^3).\nonumber
\end{align}
if we further assume the stochastic gradient function $\g(\x)$ satisfies \assum{mss} with mean-squared smoothness parameter $\bar{L}$, where the complexity measure $\mathcal{T}^{\sto}_{\epsilon}(\cdot)$ is defined in Eq.~\eqn{stochastic-complexity-measure}, the function class $\mathcal{F}_1(\Delta,\bar{L})$ is defined in \defn{Fp}. The lower bound still holds even if we restrict $\mathcal{F}_1(\Delta,\bar{L})$ to functions whose domain has dimension $\tilde{O}(\Delta\bar{L}\sigma/\epsilon^3)$.
\end{theorem}
\begin{remark}
Compared to the classical result~\cite{arjevani2022lower}, the lowest possible dimension of the hard instance is improved from $\Theta(\epsilon^{-4})$ to $\Theta(\epsilon^{-3})$, which is due to a sharper analysis and may be of independent interest.
\end{remark}
The proof of \thm{stochastic-formal-mss} is deferred to \append{stochastic-lowerbound-unbounded-mss}.

%%%%%%%%%%%%%%%%%%%%%%%%%%%%%%%%%%%%%%%%%%%%%%%%%%%%%%%%%%%%%%%%%%%%%%%%%%%%%%%%%%%%%%%%%%%%%%%%%%%%%%%

\section*{Acknowledgements}
We thank Hao Wang for helpful discussions regarding \sec{stochastic-quantum-lowerbound}, and also thank anonymous reviewers for helpful suggestions on an initial version of this paper. TL was supported by a startup fund from Peking University.

%%%%%%%%%%%%%%%%%%%%%%%%%%%%%%%%%%%%%%%%%%%%%%%%%%%%%%%%%%%%%%%%%%%%%%%%%%%%%%%%%%%%%%%%%%%%%%%%%%%%%%%

\newcommand{\arxiv}[1]{arXiv:\href{https://arxiv.org/abs/#1}{\ttfamily{#1}}\?}\newcommand{\arXiv}[1]{arXiv:\href{https://arxiv.org/abs/#1}{\ttfamily{#1}}\?}\def\?#1{\if.#1{}\else#1\fi}

%%%%%%%%%%%%%%%%%%%%%%%%%%%%%%%%%%%%%%%%%%%%%%%%%%%%%%%%%%%%%%%%%%%%%%%%%%%%%%%
%%%%%%%%%%%%%%%%%%%%%%%%%%%%%%%%%%%%%%%%%%%%%%%%%%%%%%%%%%%%%%%%%%%%%%%%%%%%%%%
% APPENDIX
%%%%%%%%%%%%%%%%%%%%%%%%%%%%%%%%%%%%%%%%%%%%%%%%%%%%%%%%%%%%%%%%%%%%%%%%%%%%%%%
%%%%%%%%%%%%%%%%%%%%%%%%%%%%%%%%%%%%%%%%%%%%%%%%%%%%%%%%%%%%%%%%%%%%%%%%%%%%%%%
\newpage
\appendix
\onecolumn
\section{Auxiliary Lemmas}\label{append:existing-lemmas}
We list the auxiliary lemmas used in our proofs here.
\begin{lemma}[\citealt{carmon2020lower}, Lemma 1]\label{lem:small-components-no-impact}
The functions $\bar{f}_T$, $\Psi$ and $\Phi$ satisfy the following.
\begin{enumerate}
\item For all $x\leq \frac{1}{2}$ and all $k\in\mathbb{N}$, $\Psi^{(k)}(x)=0$.
\item For all $x\geq 1$ and $|y|<1$, $\Psi(x)\Phi'(y)>1$.
\item $\forall\x\in\R^T$,
\begin{align}
\sqrt{\sum_{i=t}^T x_i^2}\leq\frac{1}{2}\Rightarrow \nabla^{(0,\ldots,p)}\bar{f}_T(x_1,x_2,\ldots,x_T)=\nabla^{(0,\ldots,p)}\bar{f}_T(x_1,\ldots,x_t,0,\ldots,0).\nonumber
\end{align}

\item Both $\Psi$ and $\Phi$ are infinitely differentiable, and for all $k\in\N$ we have
\begin{align}
\sup_\x|\Psi^{(k)}(\x)|\leq\exp\Big(\frac{5k}{2}\log(4k)\Big),\qquad\sup_\x|\Phi^{(k)}(\x)|\leq\exp\Big(\frac{3k}{2}\log\frac{3k}{2}\Big).\nonumber
\end{align}
\item The functions and derivatives $\Phi$, $\Psi$, $\Phi'$ and $\Phi'$ are non-negative and bounded, with
\begin{align}
0\leq\Psi<e,\quad 0\leq\Psi'\leq\sqrt{54/e},\quad 0<\Phi<\sqrt{2\pi e},\quad 0<\Phi'\leq\sqrt{e}.\nonumber
\end{align}
\end{enumerate}
\end{lemma}

\begin{lemma}[\citealt{carmon2020lower}, Lemma 2]\label{lem:fT-large-gradient}
If $|x_i|<1$ for any $i\leq T$, then there exists $j\leq i$ such that $|x_j|<1$ and
\begin{align}
\|\nabla\bar{f}_T(\x)\|\geq\Big|\frac{\partial}{\partial x_j}\bar{f}_T(\x)\Big|> 1.\nonumber
\end{align}
\end{lemma}

\begin{lemma}[\citealt{carmon2020lower}, Lemma 3]\label{lem:fT-boundedness}
The function $\bar{f}_T\colon\R^d\to\R$ defined in Eq.~\eqn{bar-fT} satisfies the following.
\begin{enumerate}
\item $\bar{f}_T(\0)-\inf_\x\bar{f}_T(\x)\leq 12T$.
\item For all $\x\in\R^d$, we have $\|\nabla\bar{f}_T(\x)\|_{\infty}\leq 23$ and $\|\nabla\bar{f}_T(\x)\|\leq 23\sqrt{T}$.
\item For all $p\geq 1$, the $p$-th order derivatives of $\bar{f}_T$ are $\ell_p$-Lipschitz continuous, where
\begin{align}
\ell_p\leq\exp\Big(\frac{5}{2}p\log p+cp\Big)\nonumber
\end{align}
for some numerical constant $c<\infty$.
\end{enumerate}

\end{lemma}

\begin{lemma}[\citealt{garg2020no}, Proposition 14]\label{lem:vector-from-sphere}
Let $\x\in\mathbb{B}(\0,1)$. Then for a $d$-dimensional random unit vector $\u$ and all $c>0$,
\begin{align}
\Pr_\u(|\<\x,\u\>|\geq c)\leq 2e^{-dc^2/2}.\nonumber
\end{align}
\end{lemma}

\begin{lemma}[\citealt{grover1996fast}, Grover's algorithm]\label{lem:grover}
For any function $\omega(\xi)\colon\Xi\to\R$ satisfying
\begin{align}
\Pr_{\xi\in\Xi}\{\omega(\xi)\neq 0\}=1-p\nonumber
\end{align}
for some $p<1$, with probability at least $1/2$ we can find a $\xi$ satisfying $\omega(\xi)\neq 0$ using $O(1/\sqrt{p})$ queries to the following quantum oracle
\begin{align}
U_{\omega}\ket{\xi}\ket{y}=\ket{\xi}\ket{y+\omega(\xi)}.\nonumber
\end{align}
\end{lemma}

\begin{lemma}[\citealt{arjevani2022lower}, Observation 1]\label{lem:Gamma-properties}
The function $\Gamma$ defined in Eq.~\eqn{Gamma-defn} and $\Theta_i$ defined in Eq.~\eqn{Theta_i-defn} satisfies\footnote{The last entry can be found in the proof of Lemma 4 of~\citet{arjevani2022lower}.}
\begin{enumerate}
\item $\Gamma(t)=0$ for all $t\in(-\infty,1/(4\beta)]$.
\item $\Gamma(t)=1$ for all $t\in[1/(2\beta),\infty)$.
\item $\Gamma\in C^{\infty}$, with $0\leq\Gamma'(t)\leq 6/\beta$ and $|\Gamma''(t)|\leq 128/\beta^2$ for all $t\in\R$.
\item $\Theta_i$ is $(36/\beta)$-Lipschitz.
\end{enumerate}

\end{lemma}

%%%%%%%%%%%%%%%%%%%%%%%%%%%%%%%%%%%%%%%%%%%%%%%%%%%%%%%%%%%%%%%%%%%%%%%%%

\section{Probabilistic Facts about $\tildef$}\label{append:probabilistic-facts}
In this subsection, we discuss some probabilistic facts of the hard instance $\tildef$ defined in Eq.~\eqn{tildef-defn} that are useful for the proof of quantum lower bounds.

\begin{lemma}\label{lem:quantum-zero-chain}
Let $1 \leq t \leq T$ be integers and $\{\u^{(1)},\ldots,\u^{(t-1)}\}$ be a set of orthonormal vectors. Let $\{\u^{(t)},\ldots, \u^{(T)}\}$ be chosen uniformly at random so that the set $\{\u^{(1)},\ldots,\u^{(T)}\}$ is orthonormal. Then
\begin{align}
\forall\x\in\mathbb{B}(\0,2\beta\sqrt{T}),\quad\Pr_{\u^{(t)},\ldots,\u^{(T)}}\big(\nabla^{(0,\ldots,p)}\tildef(\x)\neq\nabla^{(0,\ldots,p)}\tilde{f}_{t;U_t}(\x)\big)\leq \frac{1}{144T^4},\nonumber
\end{align}
where $U_t\in\R^{d\times t}$ is defined as the orthogonal matrix with columns $\u^{(1)},\ldots,\u^{(t)}$ and $\tilde{f}_{t;U_t}(\x)$ is defined in \eqn{tildeft-defn}, given that the dimension $d$ of $\tildef$ satisfies $d\geq 200T \log T$.

If $d$ further satisfies $d\geq 400T\log T$, we have
\begin{align}
\forall\x\in\mathbb{B}(\0,2\beta\sqrt{T}),\quad\Pr_{\u^{(t)},\ldots,\u^{(T)}}\big(\nabla^{(0,\ldots,p)}\tildef(\x)\neq\nabla^{(0,\ldots,p)}\tilde{f}_{t;U_t}(\x)\big)\leq \frac{1}{144T^6}.\nonumber
\end{align}
If $d$ satisfies $d\geq 400\mathscr{T}T\log\mathscr{T}$ for some $\mathscr{T}$ satisfying $\mathscr{T}\geq T$, we have
\begin{align}
\forall\x\in\mathbb{B}(\0,2\beta\sqrt{T}),\quad\Pr_{\u^{(t)},\ldots,\u^{(T)}}\big(\nabla^{(0,\ldots,p)}\tildef(\x)\neq\nabla^{(0,\ldots,p)}\tilde{f}_{t;U_t}(\x)\big)\leq \frac{1}{144\mathscr{T}^2T^4}.\nonumber
\end{align}

\end{lemma}

\begin{proof}
Without loss of generality, in this proof we set $\alpha=\beta=1$. Use $\x_{\perp}$ to denote the projection of $\x$ to the span $\{\u^{(t)},\ldots, \u^{(T)}\}$. Intuitively, as long as each component of $\x_{\perp}$ has a small absolute value, the components $\{\u^{(t)},\ldots, \u^{(T)}\}$ will have no impact on the function value and any order of derivative. Quantitatively, by \lem{small-components-no-impact} we can derive that
\begin{align}
\Pr\big(\nabla^{(0,\ldots,p)}\tildef(\x)\neq\nabla^{(0,\ldots,p)}\tilde{f}_{t;U_t}(\x)\big)\leq 1-\Pr\Big(|\<\u^{(t)},\x\>|,\ldots,|\<\u^{(T)},\x\>|\leq\frac{\beta}{2}\Big).\nonumber
\end{align}
Since $\{\u^{(t)},\ldots, \u^{(T)}\}$ are chosen uniformly at random in the $(d-t+1)$-dimensional orthogonal complement of span $\{\u^{(1)},\ldots,\u^{(t-1)}\}$, by \lem{vector-from-sphere} we can further derive that
\begin{align}
1-\Pr\Big(|\<\u^{(t)},\x\>|,\ldots,|\<\u^{(T)},\x\>|\leq\frac{\beta}{2}\Big)&\leq 1-\prod_{i=t}^{T}\big[1-\Pr\big(|\<\u^{(i)},\x\>|\geq \beta/2\big)\big]\nonumber\\
&\leq\sum_{i=t}^T\Pr\big[|\<\u^{(i)},\x\>|\geq \beta/2\big]\nonumber\\
&=2Te^{-(d-T)\cdot(4\sqrt{T})^{-2}/2}.\nonumber
\end{align}
We can then reach our desired conclusions considering different values of $d$.
\end{proof}

\begin{lemma}[Cannot guess stationary point]\label{lem:cannot-guess}
Let $k<T$ be a positive integer and $\{\u^{(1)},\ldots,\u^{(k)}\}$ be a set of orthonormal vectors. Let $\{\u^{(k+1)},\ldots,\u^{(T)}\}$ be chosen uniformly at random from $\spn(\u^{(1)},\ldots, \u^{(k)})^{\perp}$ such that all columns of the matrix $U=\big[\u^{(1)},\ldots,\u^{(T)}\big]$ forms a set of orthonormal vectors. Then,
\begin{align}
\forall \x\in\mathbb{B}(\0,2\beta\sqrt{T}),\quad\Pr_{\{\u^{(k+1)},\ldots,\u^{(T)}\}}\big[\|\nabla\tildef(\x)\|\leq\alpha/\beta\big]\leq \frac{1}{144T^4},\nonumber
\end{align}
for the function $\tildef(\x)\colon\R^d\to\R$ defined in Eq.~\eqn{tildef-defn}, given that the dimension $d$ satisfies $d\geq200T\log T$. If $d$ further satisfies $d\geq 400T\log T$, we have
\begin{align}
\forall\x\in\mathbb{B}(\0,2\beta\sqrt{T}),\quad\Pr_{\u^{(t)},\ldots,\u^{(T)}}\big(\nabla^{(0,\ldots,p)}\tildef(\x)\neq\nabla^{(0,\ldots,p)}\tilde{f}_{t;U_t}(\x)\big)\leq \frac{1}{144T^6}.\nonumber
\end{align}
\end{lemma}
\begin{proof}
For any $\x\in\mathbb{B}(\0,2\beta\sqrt{T})$, by \lem{vector-from-sphere} we can claim that the quantity
\begin{align}
\Pr_{\u^{(T)}}[\big|\<\u^{(T)},\x\>\big|\geq\beta]\leq 2\exp(-d(2\sqrt{T})^{-2}/2)\nonumber
\end{align}
is less than $\frac{1}{144T^4}$ if $d\geq 200T\log T$ and is further less than $\frac{1}{144T^6}$ if $d\geq 400T\log T$. Further by \lem{fT-large-gradient}, given the condition that $|\<\u^{(T)},\x\>\big|<\beta$, we have
\begin{align}
\|\nabla\tildef(\x)\|\geq\frac{\alpha}{\beta}\Big|\frac{\partial}{\partial x_T}\bar{f}_T(U^T\x/\beta)\Big|>\frac{\alpha}{\beta}.\nonumber
\end{align}
Hence, we can conclude that with probability at most $\frac{1}{144T^4}$ or $\frac{1}{144T^6}$ separately under the conditions $d\geq 200T\log T$ or $d\geq 400T\log T$, the following inequality is true:
\begin{align}
\|\nabla\tildef(\x)\|\leq \alpha/\beta.
\end{align}

\end{proof}

%%%%%%%%%%%%%%%%%%%%%%%%%%%%%%%%%%%%%%%%%%%%%%%%%%%%%%%%%%%%%%%%%%%%%%%%%%%%%%%

\section{Proof of Quantum Lower Bound with Lipschitz $p$-th Order Derivatives}\label{append:p-th-order-lower-bound}

\subsection{Lower Bound with Bounded Input Domain}\label{append:noiseless-lowerbound}
In this subsection, we prove a query complexity lower bound for any quantum algorithm $A_{\quan}$ defined in \sec{quantum-model} on a function class with bounded input domain using the hard instance $\tildef$ defined in Eq.~\eqn{tildef-defn}. Quantitatively, we define the function class $\tilde{\mathcal{F}}_p(\Delta, L_p,\mathcal{R})$ with bounded input domain as follows.

\begin{definition}\label{defn:tildeFp-bounded}
Let $p\geq 1$, $\Delta>0$ and $L_p>0$. Then the set $\tilde{\mathcal{F}}_p(\Delta, L_p,\mathcal{R})$ denotes the union, over $d\in\mathbb{N}$, of the collection of $C^{\infty}$ functions $f\colon\mathbb{B}(\0,\mathcal{R})\to\mathbb{R}$ with $L_p$-Lipschitz $p$-th derivative and $f(\0)-\inf_\x f(\x)\leq\Delta$.
\end{definition}
For the convenience of notations, we use $\widetilde{O}^{p}_{T;U}$ to denote the quantum evaluation oracle encoding the $p$-th-order derivatives of function $\tildef$, or equivalently
\begin{align}
\widetilde{O}^{(p)}_{T;U}\ket{\x}\ket{y}\to\ket{\x}\ket{y\oplus\nabla^{(0,\ldots,p)}\tildef(\x)}.
\end{align}
Consider the truncated sequence $A_{\quan}^{(k)}$ of any possible quantum algorithm $A_{\quan}$ with $k<T$, we define a sequence of unitaries starting with $A_0=A_{\quan}^{(k)}$ as follows:
\begin{align}
A_0&:=V_k\tildeOp_{T;U}V_{k-1}\tildeOp_{T;U}\cdots \tildeOp_{T;U}V_1\tildeOp_{T;U}V_0\nonumber \\
A_1&:=V_k\tildeOp_{T;U}V_{k-1}\tildeOp_{T;U}\cdots \tildeOp_{T;U}V_1\tildeOp_{1;U_1}V_0\nonumber \\
A_2&:=V_k\tildeOp_{T;U}V_{k-1}\tildeOp_{T;U}\cdots \tildeOp_{2;U_2}V_1\tildeOp_{1;U_1}V_0 \\
&\vdots \nonumber\\
A_k&:=V_k\tildeOp_{k;U_k}V_{k-1}\tildeOp_{k-1;U_{k-1}}\cdots \tildeOp_{2;U_2}V_1\tildeOp_{1;U_1}V_0, \nonumber
\end{align}
where $\tildeOp_{t,U_t}$ stands for the evaluation oracle of function $\tilde{f}_{t;U_t}$ and its $p$-th-order derivatives as defined in Eq.~\eqn{tildeft-defn}. Our goal to show that the algorithm $A_0$ does not solve our problem. To achieve that, we follow a similar approach shown in~\citet{garg2020no} and develop a hybrid argument in which we first show that the outputs of the algorithm $A_i$ and $A_{i+1}$ are close, so does the outputs of $A_0$ and $A_k$. Then, we argue that the algorithm $A_k$ cannot find an $\epsilon$-stationary point with high probability since oracles in the algorithm are independent from $\u_T$. Hence, $A_k$ cannot do better than random guessing a vector $\u_T$, which by \lem{cannot-guess} in \append{probabilistic-facts} fails with overwhelming probability.

\begin{proposition}\label{prop:nonstochastic-bounded}
There exist numerical constants $0<c_0,c_1<\infty$ such that the following lower bound holds. Let $p\geq 1$, $p\in \mathbb{N}$, and let $\Delta$, $L_p$, and $\epsilon$ be positive. Then,
\begin{align*}
\mathcal{T}_{\epsilon}\big(\mathcal{A}_{\quan},\tilde{\mathcal{F}}_p(\Delta,L_p,\mathcal{R})\big)
\geq c_0\Delta(L_p/\ell_p)^{1/p}\epsilon^{-\frac{1+p}{p}},
\end{align*}
where $\ell_p\leq e^{\frac{5}{2}p\log p+c_1p}$, $\mathcal{R}=\sqrt{c_0\Delta}\big(\frac{\ell_p}{L_p}\big)^{\frac{1}{2p}}\epsilon^{-\frac{p-1}{2p}}$, the complexity measure is defined in Eq.~\eqn{complexity-measure}, and the function class $\tilde{F}_p(\Delta,L_p,\mathcal{R})$ is defined in \defn{tildeFp-bounded}. The lower bound holds even if we restrict $\tilde{\mathcal{F}}_p(\Delta,L_p,\mathcal{R})$ to functions whose domain has dimension
\begin{align}
\frac{200c_0\Delta L_p^{1/p}}{\ell_p^{1/p}}\epsilon^{-\frac{1+p}{p}}\cdot\log\Big(\frac{c_0\Delta L_p^{1/p}}{\ell_p^{1/p}}\cdot\epsilon^{-\frac{1+p}{p}}\Big).\nonumber
\end{align}
\end{proposition}

The following lemma is helpful for proving \prop{nonstochastic-bounded}.
\begin{lemma}\label{lem:A_0-cannot}
Consider the $d$-dimensional function $\tildef(\x)\colon\mathbb{B}(\0,2\beta\sqrt{T})\to\R$ defined in \eqn{tildef-defn} with the rotation matrix $U$ being chosen arbitrarily and the dimension $d\geq 200T\log T$. Consider the truncated sequence $A_{\quan}^{(k)}$ of any possible quantum algorithm $A_{\quan}$ containing $k<T$ queries to the oracle $O^{(p)}_f$ defined in Eq.~\eqn{Ofp-defn}, let $p_U$ be the probability distribution over $\x\in\mathbb{B}(\0,2\beta\sqrt{T})$ obtained by measuring the state $A_{\quan}^{(t)}\ket{0}$, which is related to the rotation matrix $U$. Then,
\begin{align}
\Pr_{U,\x_{\text{out}}\sim p_U
}\big[\|\nabla\tildef(\x_{\text{out}})\|\leq\alpha/\beta\big]\leq\frac{1}{3}. \nonumber
\end{align}
\end{lemma}

\begin{proof}
We first demonstrate that $A_k$ defined in Eq.~\eqn{unitary-sequences} cannot find an $\alpha/\beta$-approximate stationary point with high probability. In particular, let $p_{U_k}$ be the probability distribution over $\x\in\mathbb{B}(\0,2\beta\sqrt{T})$ obtained by measuring the output state $A_k\ket{0}$. Then,
\begin{align}
&\Pr_{U_k,\x_{\text{out}}\in p_{U_k}}\big[\|\nabla\tildef(\x_{\text{out}})\|\leq \alpha/\beta\big]\leq \Pr_{\{\u^{(k+1)},\ldots,\u^{(T)}\}}\big[\|\nabla\tildef(\x)\|\leq \alpha/\beta\big]. \nonumber
\end{align}
for any fixed $\x$. Hence by \lem{cannot-guess},
\begin{align}
\Pr_{U_k,\x_{\text{out}}\in p_{U_k}}\big[\|\nabla\tildef(\x_{\text{out}})\|\leq\alpha/\beta\big]\leq \frac{1}{6}. \nonumber
\end{align}
Moreover, by \lem{similar-outputs} and Cauchy-Schwarz inequality, we have
\begin{align}
&\E_{U}\big[\|A_k\ket{0}-A_0\ket{0}\|^2\big]\leq k\cdot\E_U\Big[\sum_{t=1}^{k-1}\|A_{t+1}\ket{0}-A_t\ket{0}\|^2\Big]\leq\frac{1}{36T^2}. \nonumber
\end{align}
Then by Markov's inequality,
\begin{align}
\Pr_{U}\Big[\|A_k\ket{0}-A_0\ket{0}\|^2\geq\frac{1}{6T}\Big]\leq\frac{1}{6T}, \nonumber
\end{align}
since both norms are at most 1. Thus, the total variance distance between the probability distribution $p_U$ obtained by measuring $A_0\ket{0}$ and the probability distribution $p_U^{(k)}$ obtained by measuring $A_k\ket{0}$ is at most
\begin{align}
\frac{1}{6T}+\frac{1}{6T}=\frac{1}{3T}\leq\frac{1}{6}.\nonumber
\end{align}
Hence, we can conclude that
\begin{align}
\Pr_{U,\x_{\text{out}}\sim p_U
}\big[\|\nabla\tildef(\x_{\text{out}})\|\leq\alpha/\beta\big]\leq\Pr_{U_k,\x_{\text{out}}\sim p_{U_k}}\big[\|\nabla\tildef(\x_{\text{out}})\|\leq\alpha/\beta\big]+\frac{1}{6}=\frac{1}{3}.\nonumber
\end{align}
\end{proof}

Equipped with \lem{A_0-cannot}, we are now ready to prove \prop{nonstochastic-bounded}.

\begin{proof}[Proof of \prop{nonstochastic-bounded}]
We set up the scaling parameters $\alpha,\beta$ in hard instance $\tildef\colon\R^d\to\R$ defined in Eq.~\eqn{tildef-defn} for some $T$ as
\begin{align}
\alpha=\frac{L_p\beta^{p+1}}{\ell_p},\qquad\beta=\Big(\frac{\ell_p\epsilon}{L_p}\Big)^{1/p},\nonumber
\end{align}
where $\ell_p\leq e^{2.5p\log p+c_1}$ as in the third entry of \lem{fT-boundedness}. Then by \lem{fT-boundedness}, we know that the $p$-th order derivatives of $\tildef$ are $L_p$-Lipschitz continuous. Moreover, note that
\begin{align}
\tildef(\0)-\inf_{\x}\tildef(\x)&=\alpha\big(\bar{f}_T(\0)-\inf_\x\bar{f}_T(\x)\big)\leq\frac{\ell_p^{1/p}\epsilon^{\frac{1+p}{p}}}{12L_p^{1/p}}T.\nonumber
\end{align}
Then by choosing
\begin{align}
T=\frac{c_0\Delta L_p^{1/p}}{\ell_p^{1/p}}\epsilon^{-\frac{1+p}{p}},\nonumber
\end{align}
for some positive constant $c_0$, we can have $\tildef\in\tilde{\mathcal{F}}(\Delta,L_p,\mathcal{R})$ for arbitrary dimension $d$ and rotation matrix $U$. Moreover, by \lem{A_0-cannot}, for any truncated sequence $A_{\quan}^{(t)}$ of any possible quantum algorithm $A_{\quan}$ containing $t<T$ queries to the oracle $O^{(p)}_f$ on input domain $\mathbb{B}(0,\mathcal{R})$, we have
\begin{align}
&\Pr_{U,\x_{\text{out}}\sim p_U
}\big[\|\nabla\tildef(\x_{\text{out}})\|\leq\alpha/\beta\big]=\Pr_{U,\x_{\text{out}}\sim p_U
}\big[\|\nabla\tildef(\x_{\text{out}})\|\leq\epsilon\big]\leq\frac{1}{3},\nonumber
\end{align}
where $p_U$ is the probability distribution over $\x\in\mathbb{B}(\0,2\beta\sqrt{T})=\mathbb{B}(\0,\mathcal{R})$ obtained by measuring the state $A_{\quan}^{(t)}\ket{0}$, given that the dimension $d$ satisfies
\begin{align}
d&\geq 200T\log T=\frac{200c_0\Delta L_p^{1/p}}{\ell_p^{1/p}}\epsilon^{-\frac{1+p}{p}}\cdot\log\Big(\frac{c_0\Delta L_p^{1/p}}{\ell_p^{1/p}}\epsilon^{-\frac{1+p}{p}}\Big).\nonumber
\end{align}
Finally, due to \defn{quantum-complexity-measure} we can conclude that
\begin{align}
&\mathcal{T}_{\epsilon}\big(\mathcal{A}_{\quan},\tilde{\mathcal{F}}_p(\Delta,L_p,\mathcal{R})\big)\geq T=c_0\Delta\Big(\frac{L_p}{\ell_p}\Big)^{1/p}\epsilon^{-\frac{1+p}{p}}.\nonumber
\end{align}
\end{proof}

%============================================================

\subsection{Lower Bound with Unbounded Input Domain}
\label{append:thm-p-th-order-formal}

In this subsection, we extend the lower bound in \prop{nonstochastic-bounded} for bounded input domain to the setting of unbounded input domain and then present the proof of \thm{p-th-order-formal}. In particular, we consider the hard instance $\hatf$ introduced by~\citet{carmon2020lower},
\begin{align}
\hatf(\x):=\tildef(\chi(\x))+\frac{\alpha}{10}\cdot\frac{\|\x\|^2}{\beta^2},\nonumber
\end{align}
where
\begin{align}
\chi(\x):=\frac{\x}{\sqrt{1+\|\x\|^2/\hat{\mathcal{R}}^2}},\nonumber
\end{align}
with $\hat{\mathcal{R}}=230\sqrt{T}$ and $\alpha$, $\beta$, $T$ defined in Eq.~\eqn{tildef-defn}. The quadratic term $\|\x\|^2/10$ guarantees that with overwhelming probability one cannot obtain an $\epsilon$-stationary point by randomly choosing an $\x$ with large norm. In all, the constants in $\hatf$ are chosen carefully such that stationary points of $\hatf$ are in one-to-one correspondence to stationary points of the hard instance $\tildef$ concerning the setting with bounded input domain. Quantitatively,

\begin{lemma}[\citealt{carmon2020lower}, Section 5.2]\label{lem:tilde-hat-correspondence}
Let $\Delta$, $L_p$, and $\epsilon$ be positive constants. There exist numerical constants $0<c_0,c_1<\infty$ such that, under the following choice of parameters
\begin{align}
&T=\frac{c_0\Delta L_p^{1/p}}{\ell^{1/p}}\epsilon^{-\frac{1+p}{p}},\quad\alpha=\frac{L_p\beta^{p+1}}{\ell_p},\nonumber\\
&\beta=\Big(\frac{\ell_p\epsilon}{L_p}\Big)^{1/p},\quad\mathcal{R}=\sqrt{c_0\Delta}\Big(\frac{\ell_p}{L_p}\Big)^{\frac{1}{2p}}\epsilon^{-\frac{p-1}{2p}},\nonumber
\end{align}
where $\ell_p\leq e^{2.5p\log p+c_1}$ as in the third entry of \lem{fT-boundedness}, such that for any function pairs $(\tildef,\hatf)\in\tilde{\mathcal{F}}_p(\Delta,L_p,\mathcal{R})\times\mathcal{F}_p(\Delta,L_p)$ with dimension $d\geq 200T\log T$ and the same rotation matrix $U$, where the function classes are defined in \defn{Fp} and \defn{tildeFp-bounded} separately, there exists a bijection between the $\epsilon$-stationary points of $\tildef$ and the $\epsilon$-stationary points of $\hatf$ that is independent from $U$.
\end{lemma}
Equipped with \lem{tilde-hat-correspondence}, we are now ready to prove \thm{p-th-order-formal}.

\begin{proof}[Proof of \thm{p-th-order-formal}]
Note that one quantum query to the $p$-th order derivatives of $\hatf$ can be implemented by one quantum query to the $p$-th order derivatives of $\tildef$ with the same rotation $U$. Combined with \lem{tilde-hat-correspondence}, we can note that the problem of finding $\epsilon$-stationary points of $\tildef$ with unknown $U$ can be reduced to the problem of finding $\epsilon$-stationary points of $\hatf$ with no additional overhead in terms of query complexity. Then by \prop{nonstochastic-bounded}, we can conclude that
\begin{align}
&\mathcal{T}_{\epsilon}\big(\mathcal{A}_{\quan},\mathcal{F}_p(\Delta,L_p)\big)\geq \mathcal{T}_{\epsilon}\big(\mathcal{A}_{\quan},\tilde{\mathcal{F}}_p(\Delta,L_p,\mathcal{R})\big)=c_0\Delta\Big(\frac{L_p}{\ell_p}\Big)^{1/p}\epsilon^{-\frac{1+p}{p}},\nonumber
\end{align}
and the dimension dependence is the same as \prop{nonstochastic-bounded}.
\end{proof}

%%%%%%%%%%%%%%%%%%%%%%%%%%%%%%%%%%%%%%%%%%%%%%%%%%%%%%%%%%%%%%%%%%%%%%%%%%%%%%%

\section{Proof of Quantum Lower Bounds with Access to Stochastic Gradients}\label{append:stochastic-quantum-lowerbound}

%================================================================
\subsection{Overview of the Proof Techniques}
Similar to the classical stochastic lower bound result~\cite{arjevani2022lower}, we still utilize the classical hard instance defined in \citet{carmon2020lower} but with different scaling parameters. Nevertheless, the stochastic gradient function of the hard instance in \citet{arjevani2022lower} exhibits a relatively simple form. In particular, after learning the first $(t-1)$ coordinate directions, the next query would be to reveal the $t$-th coordinate direction via the component $\nabla_t f(\x)$, upon which direction \citet{arjevani2022lower} applied all the stochasticity to obtain the following stochastic gradient function
\begin{align}\label{eqn:sgf-def}
[\g(\x,\xi)]_i:=\nabla_i f(\x)\cdot\Big(1+\mathbbm{1}\{i=t\}\Big(\frac{\xi}{p}-1\Big)\Big)
\end{align}
for some probability parameter $p=O(\epsilon^2)$. Intuitively, it takes $1/p$ classical queries in expectation to reveal the gradient component $\nabla_t f$ and obtain an accurate estimation of $\nabla f$. For a quantum algorithm however, it takes only $O(1/\sqrt{p})$ queries by Grover's search algorithm~\cite{grover1996fast}, which leads to a quadratic quantum speedup, as shown in \append{special-case-quantum-speedup}.

To address this issue, inspired by the quantum lower bound on multivariate mean estimation~\cite{cornelissen2022near}, we construct a new stochastic gradient function where each stochastic gradient $\g(\x,\xi)$ all has a very small overlap with $\nabla_t f$ and one has to take at least $\Omega(1/p)$ quantum queries to obtain enough knowledge of the stochastic gradients to estimate $\nabla_t f$ and $\nabla f$ accurately. Full details regarding the construction of $\g(\x,\xi)$ are presented in \sec{stochastic-construction-to-lowerbound}. Then, a quantum lower bound can be obtained matching the existing classical algorithmic upper bound result following the same procedure as \sec{p-th-order-quantum-lowerbound}.

Furthermore, if we assume that the stochastic gradient function satisfies the mean-squared smoothness condition described in \assum{mss}, the stochastic gradient function defined in Eq.~\eqn{sgf-def} is no longer applicable as it is not continuous on certain inputs. To address this issue, we apply a similar version of the function smoothing technique introduced in \citet{arjevani2022lower} to our stochastic gradient function~\eqn{sgf-def} to obtain a ``smoothed" stochastic gradient function, as shown in \append{sgf-construction-mss}, upon which we can obtain a quantum query lower bound matching the existing classical algorithmic upper bound result given that the stochastic gradient function satisfies \assum{mss}.

%================================================================

\subsection{Quantum Speedup on the Classical Hard Instance}\label{append:special-case-quantum-speedup}

In this subsection, we present the proof of \lem{quantum-speedup-on-classical-hard}, which show that for minimizing the objective function in the hard instance for proving the classical lower bound in~\citet{arjevani2022lower}, there exists a quantum algorithm based on Grover's search algorithm that can find its stationary point using only $\tilde{O}(1/\epsilon^3)$ queries to the stochastic quantum gradient oracle, even in the absence of the mean-squared smoothness assumption. Consequently, proving the $\Omega(1/\epsilon^{4})$ quantum lower bound in \thm{stochastic-formal} necessitates the development of fundamentally new ideas. Further details regarding these ideas will be presented in the next subsection, \append{quantum-lower}.

\begin{proof}[Proof of \lem{quantum-speedup-on-classical-hard}]
Note that at any $\x\in\R^d$ such that $\prog_{1/4}(\x)=T$, one can directly reveal the exact gradient using one query to the stochastic quantum gradient oracle $O_{\g}$. As for points $\x$ with $\prog_{1/4}(\x)=t<T$, by \lem{grover} we know that, using
\begin{align}
\log(1/\delta)\cdot O(1/\sqrt{p})=O(\epsilon^{-1}\log(1/\delta))\nonumber
\end{align}
queries to the oracle $O_{\g}$, we can obtain an accurate estimation of $\nabla_t f$ and hence $\nabla f$, with success probability at least $1-\delta$. Then, we can perform gradient descent with step size being $1/L$ to reach a stationary point within $O(\epsilon^{-2})$ iterations, with overall success probability at least $1-O(\delta/\epsilon^2)$. Hence, we can complete the proof by setting $\delta=O(\epsilon^2)$, and the overall query complexity equals
\begin{align}
O(\epsilon^{-1}\log(1/\delta))\cdot O(1/\epsilon^{-2})=\tilde{O}(\epsilon^{-3}).\nonumber
\end{align}
\end{proof}

%=================================================================

\subsection{Proof of Quantum Lower Bound}\label{append:quantum-lower}

We prove \thm{stochastic-formal} in this subsection. To achieve this, we first prove that any quantum algorithm must make at least $\Omega(T)$ queries to the quantum stochastic gradient oracle to obtain an accurate estimation of the actual gradient (\lem{gradient-estimation-no-speedup}), and based on which we establish the quantum lower bound with a bounded input domain in \append{bounded-lower}. This can be further extended to the setting of unbounded input domains, consisting the proof of \thm{stochastic-formal} in \append{stochastic-lowerbound-unbounded}.

\begin{lemma}\label{lem:gradient-estimation-no-speedup}
For any $n<T/2$ and $t\leq T$, suppose in the form of \defn{quantum-SG-oracle} we are given the quantum stochastic gradient oracle $\tildeO_{\g;U}$ of $\g(\x,j)$ defined in Eq.~\eqn{quantum-hard-g}. Then for any quantum algorithm $A_{\quan}$ in the form of Eq.~\eqn{quantum-algorithm-form}, consider the sequence of unitaries $A_{\quan}^{(n)}$ truncated after the $n$ stochastic gradient oracle query
\begin{align}
A_{\quan}^{(n)}:=\tildeO_{\g;U}V_n\tildeO_{\g;U}\cdots\tildeO_{\g;U}V_2\tildeO_{\g;U}V_1,\nonumber
\end{align}
and any input state $\ket{\phi}$, we have
\begin{align}\label{eqn:quantum-SGD-ineffective}
\delta_{\perp}(n):=\E_{\{\u^{(t)},\ldots,\u^{(T)},M_\x\}}\big[\|\Pi_{t;\perp}\cdot A_{\quan}^{(n)}\ket{\phi}\|^2\big]\leq \frac{n}{18T^6},
\end{align}
where the expectation is over all possible sets $\{\u^{(t)},\ldots,\u^{(T)}\}$ and all possible set of matrices $\{M_\x\}$ at all positions $\x\in\mathbb{B}(\0,\beta\sqrt{T})$ satisfying Eq.~\eqn{M-construction}, given that the dimension $d$ of the objective function $\tilde{f}_{T;U}$ satisfies $d\geq 2T^2+T$.
\end{lemma}

\begin{proof}
We use induction to prove this claim. Firstly for $n=1$, we have
\begin{align}
\delta_{\perp}(1)&=\E_{\{\u^{(t)},\ldots,\u^{(T)},M_\x\}}\big[\|\Pi_{t;\perp}\cdot \tildeO_{\g;U}V_0\ket{\phi}\|^2\big]\nonumber\\
&=\E_{\{\u^{(t)},\ldots,\u^{(T)},M_\x\}}\big[\|\Pi_{t;\perp}\cdot \tildeO_{\g;U}\ket{\phi}\|^2\big]\nonumber\\
&\leq\E_{\{\u^{(t)},\ldots,\u^{(T)},M_\x\}}\big[\big\|\Pi_{t;\perp}\cdot \tildeO_{\g;U}\ket{\phi_{\parallel}}\big\|^2\big]+\E_{\{\u^{(t)},\ldots,\u^{(T)},M_\x\}}\big[\big\|\ket{\phi_{\perp}}\big\|^2\big]\nonumber,
\end{align}
where $\ket{\phi_{\parallel}}:=\Pi_{t;\parallel}\ket{\phi}$ and $\ket{\phi_{\perp}}:=\Pi_{t;\perp}\ket{\phi}$. Since for all components in the (possibly superposition) state $\Pi_{T;\perp}\ket{\psi}$ all the stochastic gradients have no overlap with $\{\u^{t+2},\ldots,\u^{T}\}$, by \lem{multivariate-mean-estimation}, we have
\begin{align}
\E_{\{\u^{(t)},\ldots,\u^{(T)},M_\x\}}\big[\big\|\Pi_{t;\perp}\cdot \tildeO_{\g;U}\ket{\phi_{\parallel}}\big\|^2\big]\leq\exp(-\zeta T),\nonumber
\end{align}
where $\zeta$ is a small enough constant. Moreover, by \lem{quantum-zero-chain} we have
\begin{align}
\E_{\{\u^{(t)},\ldots,\u^{(T)},M_\x\}}\big[\big\|\ket{\phi_{\perp}}\big\|^2\big]\leq \frac{1}{36T^6}.\nonumber
\end{align}
Hence,
\begin{align}
\delta_{\perp}(1)\leq\exp(-\zeta T)+\frac{1}{36T^6}\leq \frac{1}{18T^6}.\nonumber
\end{align}

Suppose the inequality \eqn{quantum-SGD-ineffective} holds for all $n\leq\tilde{n}$ for some $\tilde{n}<\frac{T}{2}$. Then for $n=\tilde{n}+1$, we denote
\begin{align}
\ket{\phi_{\tilde{n}}}:=\tildeO_{\g;U}V_{\tilde{n}-1}\tildeO_{\g;U}\cdots\tildeO_{\g;U}V_1\tildeO_{\g;U}V_0\ket{\phi}.\nonumber
\end{align}
Then,
\begin{align}
\delta_{\perp}(\tilde{n}+1)
&=\E_{\{\u^{(t)},\ldots,\u^{(T)},M_\x\}}\big[\|\Pi_{t;\perp}\cdot \tildeO_{\g;U}V_{\tilde{n}}\ket{\phi_{\tilde{n}}}\|^2\big]\nonumber\\
&\leq\E_{\{\u^{(t)},\ldots,\u^{(T)},M_\x\}}\big[\big\|\Pi_{t;\perp}\cdot \tildeO_{\g;U}V_{\tilde{n}}\ket{\phi_{\tilde{n};\parallel}}\big\|^2\big]
+\E_{\{\u^{(t)},\ldots,\u^{(T)},M_\x\}}\big[\big\|\ket{\phi_{\tilde{n};\perp}}\big\|^2\big]\nonumber\\
&\leq\E_{\{\u^{(t)},\ldots,\u^{(T)},M_\x\}}\big[\big\|\Pi_{t;\perp}\cdot\tildeO_{\g;U}V_{\tilde{n}}\ket{\phi_{\tilde{n};\parallel}}\big\|^2\big]
+\delta_{\perp}(\tilde{n}).\nonumber
\end{align}

Consider the following sequence
\begin{align}
\tildeO_{\g;U}V_{\tilde{n}}\tildeO_{\g;U}\cdots\tildeO_{\g;U}V_0\ket{\phi'}=\tildeO_{\g;U}V_{\tilde{n}}\ket{\phi_{\tilde{n};\parallel}},\nonumber
\end{align}
note that it contains $\tilde{n}+1\leq\frac{T}{2}$ queries to the stochastic gradient oracle, and at each query except the last one, the input state has no overlap with the desired space $W_{t;\perp}$. Then by \lem{multivariate-mean-estimation}, we have
\begin{align}
\E_{\{\u^{(t)},\ldots,\u^{(T)},M_\x\}}\big[\big\|\Pi_{t;\perp}\cdot\tildeO_{\g;U}V_{\tilde{n}}\ket{\phi_{\tilde{n};\parallel}}\big\|^2\big]\leq \exp(-\zeta T)+\frac{1}{36T^6}\leq \frac{1}{18T^6}.\nonumber
\end{align}
Hence, the inequality \eqn{quantum-SGD-ineffective} also holds for $n=\tilde{n}+1$.
\end{proof}

\subsubsection{Lower Bound with Bounded Input Domain}\label{append:bounded-lower}

Through this construction of quantum stochastic gradient oracle, we can prove the query complexity lower bound for any quantum algorithm $A_{\quan}$ defined in \sec{quantum-model} using the hard instance $\tildef$ defined in Eq.~\eqn{tildef-defn}. For the convenience of notations, we use $\widetilde{O}_{\g;U}$ to denote the stochastic gradient oracle defined in Eq.~\eqn{quantum-hard-g} of function $\tildef$. Similar to \sec{noiseless-lowerbound}, we consider the truncated sequence $A_{\quan}^{(K\cdot T/2)}$ of any possible quantum algorithm $A_{\quan}$ with $K<T$, and define a sequence of unitaries starting with $A_0=A_{\quan}^{(K\cdot T/2)}$ as follows:
\begin{align}\label{eqn:stochastic-unitary-sequences}
A_0&:=V_{K+1}\tildeO_{\g;U}V_{K;T/2}\cdots\tildeO_{\g;U}V_{K;1}\cdots\tildeO_{\g;U}V_{2;T/2}\cdots\tildeO_{\g;U}V_{2;1}\tildeO_{\g;U}V_{1;T/2}\cdots\tildeO_{\g;U}V_{1;1}\\ \nonumber
A_1&:=V_{K+1}\tildeO_{\g;U}V_{K;T/2}\cdots\tildeO_{\g;U}V_{K;1}\cdots\tildeO_{\g;U}V_{2;T/2}\cdots\tildeO_{\g;U}V_{2;1}\tildeO_{\g;U_1}V_{1;T/2}\cdots\tildeO_{\g;U_1}V_{1;1}\\ \nonumber
A_2&:=V_{K+1}\tildeO_{\g;U}V_{K;T/2}\cdots\tildeO_{\g;U}V_{K;1}\cdots\tildeO_{\g;U_2}V_{2;T/2}\cdots\tildeO_{\g;U_2}V_{2;1}\tildeO_{\g;U_1}V_{1;T/2}\cdots\tildeO_{\g;U_1}V_{1;1}\\ \nonumber
&\vdots\\ \nonumber
A_K&:=V_{K+1}\tildeO_{\g;U_K}V_{K;T/2}\cdots\tildeO_{\g;U_K}V_{K;1}\cdots\tildeO_{\g;U_2}V_{2;T/2}\cdots\tildeO_{\g;U_2}V_{2;1}\tildeO_{\g;U_1}V_{1;T/2}\cdots\tildeO_{\g;U_1}V_{1;1},
\end{align}
where $\tildeO_{\g;U_t}$ stands for the stochastic gradient oracle of function $\tilde{f}_{t;U_t}$. Note that for the sequence of unitaries $A_0$, it can be decomposed into the product of $V_{K+1}$ and $K$ sequences of unitaries, each of the form
\begin{align}
\mathscr{A}_k(n)=\tildeO_{\g;U}V_{k;n}\tildeO_{\g;U}\cdots\tildeO_{\g;U}V_{k;2}\tildeO_{\g;U}V_{k;1}\nonumber
\end{align}
for $n=T/2$ and $k\in[K]$ for some unitaries $V_1,\ldots,V_n$. In the following lemma we demonstrate that, for such sequence $\mathscr{A}_k(n)$, if we replace $\tilde{O}_{\g;U}$ by another oracle that only reveals part information of $f$, the sequence will barely change on random inputs.

\begin{lemma}\label{lem:part-sequences-closeness}
For any $t\in[T-1]$ and any $n\leq\frac{T}{2}$, consider the following two sequences of unitaries
\begin{align}
\mathscr{A}(n)=\tildeO_{\g;U}V_n\tildeO_{\g;U}\cdots\tildeO_{\g;U}V_2\tildeO_{\g;U}V_1,\nonumber
\end{align}
and
\begin{align}
\hat{\mathscr{A}}_t(n)=\tildeO_{\g;U_t}V_n\tildeO_{\g;U_t}\cdots\tildeO_{\g;U_t}V_2\tildeO_{\g;U_t}V_1,\nonumber
\end{align}
we have
\begin{align}\label{eqn:SGD-closeness}
\delta(n):=\E_{\{\u^{(t)},\ldots,\u^{(T)},M_\x\}}\big[\|(\hat{\mathscr{A}}_t(n)-\mathscr{A}(n))\ket{\psi}\|^2\big]\leq\frac{n}{36T^5},
\end{align}
for any fixed pure state $\ket{\psi}$.
\end{lemma}
\begin{proof}
We use induction to prove this claim. Firstly for $n=1$, we have
\begin{align}\label{eqn:SGD-closeness-bounded}
&\E_{\{\u^{(t)},\ldots,\u^{(T)},M_\x\}}\big[\big\|(\hat{\mathscr{A}}_t(n)-\mathscr{A}(n))\ket{\psi}\big\|^2\big]\nonumber\\
&\qquad=\E_{\{\u^{(t)},\ldots,\u^{(T)},M_\x\}}\big[\big\|(\tildeO_{\g;U}-\tildeO_{\g;U_t})\ket{\psi}\big\|^2\big]\leq \frac{1}{36T^6},
\end{align}
where the last inequality follows from \lem{quantum-zero-chain}. Suppose the inequality \eqn{SGD-closeness-bounded} holds for all $n\leq\tilde{n}$ for some $\tilde{n}<\frac{T}{2}$. Then for $n=\tilde{n}+1$, we have
\begin{align}
\delta(\tilde{n}+1)
&=\E_{\{\u^{(t)},\ldots,\u^{(T)},M_\x\}}\big[\|(\hat{\mathscr{A}}_t(n)-\mathscr{A}(n))\ket{\psi}\|^2\big]\nonumber\\
&\leq\E_{\{\u^{(t)},\ldots,\u^{(T)},M_\x\}}\big[\|(\tildeO_{\g;U}-\tildeO_{\g;U_t})\ket{\psi_t}\|^2\big]+\delta(\tilde{n}),\nonumber
\end{align}
where
\begin{align}
\ket{\psi_t}=V_{\tilde{n}}\tildeO_{\g;U_t}\cdots\tildeO_{\g;U_t}V_1\ket{\psi}\nonumber
\end{align}
is a function of $U_t$ obtained by $\tilde{n}$ queries to $\tilde{O}_{\g;U_t}$. By \lem{gradient-estimation-no-speedup}, we have
\begin{align}
\E_{\{\u^{(t)},\ldots,\u^{(T)},M_x\}}[\|\Pi_{t;\perp}\ket{\psi_t}\|^2]\leq\frac{n}{18T^6}\leq\frac{1}{36T^5},\nonumber
\end{align}
indicating that $\ket{\psi_t}$ only has a very little overlap with the subspace $W_{t;\perp}$ defined in \eqn{W_t_perp-defn}, outside of which the columns $\{\u^{(t)},\ldots,\u^{(T)}\}$ of $U$ has no impact on the function value and derivatives of $\tilde{f}_{T;U}$. Thus,
\begin{align}
\delta(\tilde{n}+1)
&\leq\E_{\{\u^{(t)},\ldots,\u^{(T)},M_\x\}}\big[\|(\tildeO_{\g;U}-\tildeO_{\g;U_t})\ket{\psi_t}\|\big]+\delta(\tilde{n})\nonumber\\
&\leq\E_{\{\u^{(t)},\ldots,\u^{(T)},M_x\}}[\|\Pi_{t;\perp}\ket{\psi_t}\|^2]+\delta(\tilde{n})\leq\frac{\tilde{n}+1}{36T^5},\nonumber
\end{align}
indicating that Eq.~\eqn{SGD-closeness} also holds for $n=\tilde{n}+1$.
\end{proof}

\begin{lemma}[$A_t$ and $A_{t-1}$ have similar outputs]
\label{lem:similar-outputs-stochastic}
For a hard instance $\tildef(\x)\colon\R^d\to\R$ defined on $\mathbb{B}(\0,2\beta\sqrt{T})$ with $d\geq 2T^2+T$, let $A_t$ for $t\in[K]$ be the sequence unitaries defined in Eq.~\eqn{stochastic-unitary-sequences}. Then
\begin{align}
\E_{\{\u^{(t)},\ldots,\u^{(T)},M_\x\}}\big(\|A_t\ket{\0}-A_{t-1}\ket{\0}\|^2\big)\leq\frac{1}{72T^4}.\nonumber
\end{align}
\end{lemma}

\begin{proof}
From the definition of the unitaries in Eq.~\eqn{stochastic-unitary-sequences}, we have
\begin{align}
\|A_t\ket{\0}-A_{t-1}\ket{\0}\|=\|(\mathscr{A}(T/2)-\hat{\mathscr{A}}_t(T/2))\ket{\psi}\|,\nonumber
\end{align}
for some fixed quantum state $\ket{\psi}$ dependent on the vectors $\{\u^{(1)},\ldots,\u^{(t-1)}\}$, where
\begin{align}
\mathscr{A}(T/2)=\tildeO_{\g;U}V_{T/2}\tildeO_{\g;U}\cdots\tildeO_{\g;U}V_2\tildeO_{\g;U}V_1,\nonumber
\end{align}
and
\begin{align}
\hat{\mathscr{A}}_t(T/2)=\tildeO_{\g;U_t}V_{T/2}\tildeO_{\g;U_t}\cdots\tildeO_{\g;U_t}V_2\tildeO_{\g;U_t}V_1.\nonumber
\end{align}
By \lem{part-sequences-closeness}, we have
\begin{align}
\E_{\{\u^{(t)},\ldots,\u^{(T)},M_{\x}\}}\big(\|(\mathscr{A}(T/2)-\hat{\mathscr{A}}_t(T/2))\ket{\psi}\|^2\big)\leq\frac{1}{36T^5}\cdot\frac{T}{2}=\frac{1}{72T^4},\nonumber
\end{align}
which leads to
\begin{align}
\E_{\{\u^{(t)},\ldots,\u^{(T)},M_\x\}}\big(\|A_t\ket{\0}-A_{t-1}\ket{\0}\|^2\big)\leq\frac{1}{72T^4}.\nonumber
\end{align}
\end{proof}

\begin{proposition}\label{prop:stochastic-A_0-cannot}
Consider the $d$-dimensional function $\tildef(\x)\colon\mathbb{B}(\0,2\beta\sqrt{T})\to\R$ defined in \eqn{tildef-defn} with the rotation matrix $U$ being chosen arbitrarily and the dimension $d\geq 2T^2+T$. Consider the truncated sequence $A_{\quan}^{(K\cdot T/2)}$ of any possible quantum algorithm $A_{\quan}$ containing $KT/2$ queries to the quantum stochastic gradient oracle $\tildeO_{\g;U}$ of $\g(\x,j)$ defined in Eq.~\eqn{quantum-hard-g} with $K<T$, let $p_U$ be the probability distribution over $\x\in\mathbb{B}(\0,2\beta\sqrt{T})$ obtained by measuring the state $A_{\quan}^{(K\cdot T/2)}\ket{0}$, which is related to the rotation matrix $U$. Then,
\begin{align}
\Pr_{U,\{M_\x\},\x_{\text{out}}\sim p_U
}\big[\|\nabla\tildef(\x_{\text{out}})\|\leq\alpha/\beta\big]\leq\frac{1}{3},\nonumber
\end{align}
where the probability is subject to all possible orthogonal rotation matrices $U$, and all possible matrices $\{M_{\x}\}$ in the quantum stochastic gradient function $\g(\x,j)$ for any $\x$.
\end{proposition}

\begin{proof}
We first demonstrate that the sequence of unitaries $A_K$ defined in Eq.~\eqn{stochastic-unitary-sequences} cannot find an $\alpha/\beta$-approximate stationary point with high probability. In particular, let $p_{U_K}$ be the probability distribution over $\x\in\mathbb{B}(\0,2\beta\sqrt{T})$ obtained by measuring the output state $A_K\ket{0}$. Then,
\begin{align}
\Pr_{U_K,\{M_\x\},\x_{\text{out}}\in p_{U_K}}\big[\|\nabla\tildef(\x_{\text{out}})\|\leq \alpha/\beta\big]
\leq \Pr_{\{\u^{(k+1)},\ldots,\u^{(T)}\}}\big[\|\nabla\tildef(\x)\|\leq \alpha/\beta\big].\nonumber
\end{align}
for any fixed $\x$. Then by \lem{cannot-guess} we have
\begin{align}
\Pr_{U_K,{M_\x},\x_{\text{out}}\in p_{U_K}}\big[\|\nabla\tildef(\x_{\text{out}})\|\leq\alpha/\beta\big]\leq \frac{1}{6}.\nonumber
\end{align}
Moreover, by \lem{similar-outputs-stochastic} and Cauchy-Schwarz inequality, we have
\begin{align}
\E_{U}\big[\|A_K\ket{0}-A_0\ket{0}\|^2\big]\leq K\cdot\E_U\Big[\sum_{t=1}^{K-1}\|A_{t+1}\ket{0}-A_t\ket{0}\|^2\Big]\leq\frac{1}{72T^2}.\nonumber
\end{align}
Then by Markov's inequality,
\begin{align}
\Pr_{U}\Big[\|A_{K-1}\ket{0}-A_0\ket{0}\|^2\geq\frac{1}{6T}\Big]\leq\frac{1}{6T},\nonumber
\end{align}
since both norms are at most 1. Thus, the total variance distance between the probability distribution $p_U$ obtained by measuring $A_0\ket{0}$ and the probability distribution $p_{U_K}$ obtained by measuring $A_K\ket{0}$ is at most
\begin{align}
\frac{1}{6T}+\frac{1}{6T}=\frac{1}{3T}\leq\frac{1}{6}.\nonumber
\end{align}
Hence, we can conclude that
\begin{align}
&\Pr_{U,\{M_{\x}\},\x_{\text{out}}\sim p_U^{(t)}
}\big[\|\nabla\tildef(\x_{\text{out}})\|\leq\alpha/\beta\big]\nonumber\\
&\qquad\quad\leq\Pr_{U_K,\{M_{\x}\},\x_{\text{out}}\sim p_{U_K}}\big[\|\nabla\tildef(\x_{\text{out}})\|\leq\alpha/\beta\big]+\frac{1}{6}=\frac{1}{3}.\nonumber
\end{align}
\end{proof}

\begin{proposition}\label{prop:stochastic-bounded}
Suppose $\Delta$, $L$, $\sigma$, and $\epsilon$ are positive. Then,
\begin{align}
\mathcal{T}^{\sto}_{\epsilon}\big(\mathcal{A}_{\quan},\tilde{\mathcal{F}_1}(\Delta,L,\mathcal{R}),\sigma\big)=\Omega\Big(\frac{\min\{\Delta^2L^2,\sigma^4\}}{\epsilon^4}\Big),\nonumber
\end{align}
where $\mathcal{R}=c\cdot\min\{\sqrt{L\Delta},\sigma\}$ for some constant $c$, the complexity measure $\mathcal{T}^{\sto}_{\epsilon}(\cdot)$ is defined in Eq.~\eqn{stochastic-complexity-measure}, and the function class $\tilde{\mathcal{F}}_1(\Delta,L,\mathcal{R})$ is defined in \defn{tildeFp-bounded}. The lower bound holds even if we restrict $\tilde{\mathcal{F}}_1(\Delta,L,\mathcal{R})$ to functions whose domain has dimension
\begin{align}
\Theta\Big(\frac{\min\{L^2\Delta^2,\sigma^4\}}{\epsilon^4}\Big).\nonumber
\end{align}
\end{proposition}

\begin{proof}
We set up the scaling parameters $\alpha$ and $\beta$ in the hard instance $\tildef\colon\R^d\to\R$ defined in Eq.~\eqn{tildef-defn} as
\begin{align}
\alpha=\frac{L\beta^{2}}{\ell},\qquad\beta=\frac{2\ell\epsilon}{L},\nonumber
\end{align}
where $\ell$ is the gradient Lipschitz constant of $\bar{f}_T$ whose value is given in \lem{fT-boundedness}. We also set the parameter
\begin{align}
T=\min\Big\{\frac{\Delta\ell}{12L\beta^2},\frac{\sigma^2\beta^2}{4\gamma^2\alpha^2}\Big\}=\Theta\Big(\frac{\min\{L\Delta,\sigma^2\}}{\epsilon^2}\Big).\nonumber
\end{align}
Then by \lem{fT-boundedness}, we know that $\tildef$ is $L$-smooth, and
\begin{align}
\tildef(\0)-\inf_{\x}\tildef(\x)=\alpha\big(\bar{f}_T(\0)-\inf_\x\bar{f}_T(\x)\big)\leq\frac{12L\beta^2}{\ell}\cdot T\leq\Delta,\nonumber
\end{align}
indicating that $\tildef\in\tilde{\mathcal{F}}(\Delta,L_p,\mathcal{R})$ for arbitrary dimension $d$ and rotation matrix $U$. Moreover, at every $\x$, we have
\begin{align}
\mathbb{E}\big[\|\g(\x,j)-\nabla\tildef(\x)\|^2\big]\leq 4\alpha^2\gamma^2T/\beta^2\leq\delta^2,\nonumber
\end{align}
indicating that the variance of the stochastic gradient function $\g$ defined in Eq.~\eqn{quantum-hard-g} is bounded by $\sigma^2$. Further, we notice that the radius
\begin{align}
\mathcal{R}=2\beta \sqrt{T}=c\cdot\min\big\{\sqrt{L\Delta},\sigma\big\}\nonumber
\end{align}
for some constant $c$.

By \prop{stochastic-A_0-cannot}, for any truncated sequence $A_{\quan}^{(KT/2)}$ of any possible quantum algorithm $A_{\quan}$ containing $KT/2<T^2/2$ queries to the oracle $O^{(p)}_f$ on input domain $\mathbb{B}(0,\mathcal{R})$, we have
\begin{align}
\Pr_{U,\x_{\text{out}}\sim p_U
}\big[\|\nabla\tildef(\x_{\text{out}})\|\leq\alpha/\beta\big]=\Pr_{U,\x_{\text{out}}\sim p_U
}\big[\|\nabla\tildef(\x_{\text{out}})\|\leq\epsilon\big]\leq\frac{1}{3},\nonumber
\end{align}
where $p_U$ is the probability distribution over $\x\in\mathbb{B}(\0,2\beta\sqrt{T})=\mathbb{B}(\0,\mathcal{R})$ obtained by measuring the state $A_{\quan}^{(KT/2)}\ket{0}$, given that the dimension $d$ satisfies
\begin{align}
d\geq 2T^2+T=\Theta\Big(\frac{\min\{L^2\Delta^2,\sigma^4\}}{\epsilon^4}\Big).\nonumber
\end{align}
Then according to \defn{quantum-complexity-measure}, we can conclude that
\begin{align}
\mathcal{T}_{\epsilon}^{\sto}\big(\mathcal{A}_{\quan},\tilde{\mathcal{F}}_1(\Delta,L,\mathcal{R}),\sigma\big)
\geq \frac{T^2}{2}
=\Omega\Big(\frac{\min\{L^2\Delta^2,\sigma^4\}}{\epsilon^4}\Big).\nonumber
\end{align}
\end{proof}

\subsubsection{Lower Bound with Unbounded Input Domain}\label{append:stochastic-lowerbound-unbounded}
In this subsection, we extend the quantum lower bound proved in \prop{stochastic-bounded} to the function class $\mathcal{F}_1(\Delta,L)$ with unbounded input domain via similar scaling techniques adopted in~\citet{arjevani2022lower} and \append{thm-p-th-order-formal}. In particular, we consider the scaled hard instance $\hatf$ introduced in~\citet{carmon2020lower} and also used in~\citet{arjevani2022lower},
\begin{align}
\hatf(\x):=\tildef(\chi(\x))+\frac{\alpha}{10}\cdot\frac{\|\x\|^2}{\beta^2},\nonumber
\end{align}
where
\begin{align}
\chi(\x):=\frac{\x}{\sqrt{1+\|\x\|^2/\hat{\mathcal{R}}^2}},\nonumber
\end{align}
with parameters
\begin{align}
\alpha=\frac{L\beta^{2}}{\ell},\qquad\beta=\frac{2\ell\epsilon}{L},\qquad\beta=\frac{2\ell\epsilon}{L},\qquad T=\min\Big\{\frac{\Delta\ell}{12L\beta^2},\frac{\sigma^2\beta^2}{4\gamma^2\alpha^2}\Big\},\qquad \hat{\mathcal{R}}=230\sqrt{T},\nonumber
\end{align}
whose values are also adopted in the proof of \prop{stochastic-bounded}. The constants in $\hatf$ are chosen carefully such that stationary points of $\hatf$ are in one-to-one correspondence to stationary points of the hard instance $\tildef$ concerning the setting with bounded input domain. Quantitatively,

\begin{lemma}[\citealt{arjevani2022lower}]\label{lem:stochastic-tilde-hat-correspondence}
Let $\Delta$, $L_p$, and $\epsilon$ be positive constants. There exist numerical constants $0<c_0,c_1<\infty$ such that, under the following choice of parameters
\begin{align}
T=\min\Big\{\frac{\Delta\ell}{12L\beta^2},\frac{\sigma^2\beta^2}{4\gamma^2\alpha^2}\Big\},\qquad
\alpha=\frac{L\beta^{2}}{\ell},\qquad\beta=\frac{2\ell\epsilon}{L},\qquad
\mathcal{R}=c\sqrt{T}\cdot\min\{\sqrt{L\Delta},\sigma\},\nonumber
\end{align}
where $\ell$ is the gradient Lipschitz parameter of $\bar{f}_T$ whose value is given in \lem{fT-boundedness}, such that for any function pairs $(\tildef,\hatf)\in\tilde{\mathcal{F}}_1(\Delta,L,\mathcal{R})\times\mathcal{F}_1(\Delta,L)$ with dimension $d\geq 400T\log T$ and the same rotation matrix $U$, where the function classes are defined in \defn{Fp} and \defn{tildeFp-bounded} separately, there exists a bijection between the $\epsilon$-stationary points of $\tildef$ and the $\epsilon$-stationary points of $\hatf$ that is independent from $U$.
\end{lemma}
Equipped with \lem{stochastic-tilde-hat-correspondence}, we are now ready to prove \thm{stochastic-formal}.

\begin{proof}[Proof of \thm{stochastic-formal}]
Note that one quantum query to the stochastic gradient of $\hatf$ can be implemented by one quantum query to the stochastic gradient of $\tildef$ with the same rotation $U$, if we directly scale the stochastic gradient function of $\tildef$ to $\hatf$, which will not increase the variance of the stochastic gradient function. Combined with \lem{stochastic-tilde-hat-correspondence}, we can note that the problem of finding $\epsilon$-stationary points of $\tildef$ with unknown $U$ can be reduced to the problem of finding $\epsilon$-stationary points of $\hatf$ with no additional overhead in terms of query complexity. Then by \prop{stochastic-bounded}, we can conclude that
\begin{align}
&\mathcal{T}^{\sto}_{\epsilon}\big(\mathcal{A}_{\quan},\mathcal{F}_1(\Delta,L),\sigma\big)\geq \mathcal{T}^{\sto}_{\epsilon}\big(\mathcal{A}_{\quan},\tilde{\mathcal{F}}_1(\Delta,L,\mathcal{R}),\sigma\big)=\Omega\Big(\frac{\max\{L^2\Delta^2,\sigma^4\}}{\epsilon^4}\Big),\nonumber
\end{align}
and the dimension dependence is the same as \prop{stochastic-bounded}.
\end{proof}

%====================================================================================================================

\subsection{Proof of Quantum Lower Bound with the Mean-Squared Smoothness Assumption}\label{append:mss}

In this subsection, we prove a quantum query lower bound for finding an $\epsilon$-stationary point with access to the quantum stochastic gradient oracle defined in \defn{quantum-SG-oracle} and additionally satisfies the \textit{mean-squared smoothness} assumption defined in \assum{mss} for some constant $\bar{L}$.

%====================================================================================

\subsubsection{Construction of the Stochastic Gradient Function Satisfying \assum{mss}}\label{append:sgf-construction-mss}

Note that the stochastic gradient function~\eqn{quantum-hard-g} in \sec{stochastic-construction-to-lowerbound} does not satisfy \assum{mss} since the function $\prog_\alpha(\cdot)\colon\R^d\to\R$ defined in~\eqn{defn-prog} contains a maximization over all the $d$ components, which makes the stochastic gradient discontinuous.

This issue can be addressed using a smoothing technique similar to which introduced in \citet{arjevani2022lower}. In particular, \citet{arjevani2022lower} defines the following smoothed version of the indicator function $\mathbb{I}\{i>\prog_{\frac{\beta}{4}}(\x)\}$ for any $i$ (with rotation $U$):
\begin{align}\label{eqn:Theta_i-defn}
\Theta_i(\x)\coloneqq\Gamma\bigg(1-\Big(\sum_{k=i}^T\Gamma^2(|x_k/\beta|)\Big)^{1/2}\bigg)=\Gamma(1-\|\Gamma(\x_{\geq i})\|),
\end{align}
where $\Gamma(|\x_{\geq i}|)$ is a shorthand for a vector with entries
\begin{align}
\Gamma(|x_i|),\Gamma(|x_{i+1}|),\ldots,\Gamma(|x_T|),\nonumber
\end{align}
and the function $\Gamma\colon\R\to\R$ is defined as
\begin{align}\label{eqn:Gamma-defn}
\Gamma(t)=\frac{\int_{1/4}^{t/\beta}\Lambda(\tau)\d\tau}{\int_{1/4}^{1/2}\Lambda(\tau)\d\tau},\quad\text{where}\quad
\Lambda(t)=
\begin{cases}
0,&\frac{t}{\beta}\leq\frac{1}{4}\text{ or }\frac{t}{\beta}\geq\frac{1}{2}, \\
\exp\Big(-\frac{1}{100\big(\frac{t}{\beta}-\frac{1}{4}\big)\big(\frac{1}{2}-\frac{t}{\beta}\big)}\Big), &\frac{1}{4}<\frac{t}{\beta}<\frac{1}{2}.
\end{cases}
\end{align}
Note that $\Gamma$ is a smooth non-decreasing Lipschitz function with $\Gamma(t)=0$ for all $t\leq \beta/4$ and $\Gamma(t)=1$ for all $t\geq \beta/2$. Then, the function $\Theta_i(\x)$ defined in Eq.~\eqn{Theta_i-defn} satisfies
\begin{align}
\mathbb{I}\left\{i>\prog_{\frac{\beta}{4}}(\x)\right\}\leq\Theta_i(\x)\leq\mathbb{I}\left\{i>\prog_{\frac{\beta}{2}}(\x)\right\}.\nonumber
\end{align}

Following the same intuition of the gradient function defined in Eq.~\eqn{quantum-hard-g} without the mean-squared smoothness assumption, here we also arrange the stochasticity to harden the attempts on increasing the coordinate progress via stochastic gradient information. In particular, similar to Eq.~\eqn{M-construction}, for the $d$-dimensional function $\tilde{f}_{T;U}$ with $d\geq 4\mathscr{T}$ for some integer $\mathscr{T}$ whose value is specified later, we note that for any point $\x$ with gradient $\g(\x)$ there exists a matrix $M_\x\in\R^{d\times 2\mathscr{T}}$ with $\mathscr{T}$ columns being $\0$ and the other $\mathscr{T}$ columns forming a set of orthonormal vectors such that
\begin{align}\label{eqn:M-construction-mss}
\nabla_{\prog_{\frac{\beta}{2}}(\x)+1}\tildef(\x)=\frac{1}{2\mathscr{T}}\sum_j 2\gamma\sqrt{\mathscr{T}}\cdot\vect{m}_{\x}^{(j)},
\end{align}
where $\vect{m}_{\x}^{(j)}$ stands for the $j$-th column of $M_\x$ and
\begin{align}
\gamma=\big\|\nabla_{\prog_{\frac{\beta}{2}}(\x)+1}\tildef(\x)\big\|\leq 23\nonumber
\end{align}
is the norm of the $(\prog_{\beta/2}+1)$-th gradient component at certain points whose exact value is specified later.

Moreover, to guarantee that all the stochastic gradients at $\x$ can only reveal the $(\prog_{\beta/4}(\x)+1)$-th coordinate direction $\u_{\prog_{\beta/4}(\x)+1}$ even with infinite number of queries and will not ``accidentally" make further progress, we additionally require that for any $\x,\y\in\R^d$ with $\prog_{\beta/4}(\x)\neq\prog_{\beta/4}(\y)$, all the columns of $M_\x$ are orthogonal to all the columns of $M_\y$. This can be achieved by creating $T$ orthogonal subspaces
\begin{align}
\{\mathcal{V}_1,\ldots,\mathcal{V}_T\},\nonumber
\end{align}
where each subspace is of dimension $2T$ and has no overlap with $\{\u_1,\ldots,\u_T\}$, such that for any $\x$ the columns of $M_\x$ are within the subspace
\begin{align}
\spn\big\{\u_{\prog_{\frac{\beta}{4}}(\x)+1},\mathcal{V}_{\prog_{\frac{\beta}{4}}(\x)+1}\big\},\nonumber
\end{align}
as long as the dimension $d$ is larger than $2\mathscr{T}T+T=O(\mathscr{T}T)$.

Now, we can define the following stochastic gradient function for $\nabla\tildef(\x)$:
\begin{align}\label{eqn:quantum-hard-g-mss}
\hat{\g}(\x,j)=\g(\x)+\Theta_{\prog_{\beta/2}(\x)+1}(\x)\cdot\big(2\gamma\sqrt{\mathscr{T}}\cdot\vect{m}^{(j)}-\g_{\prog_{\beta/2}(\x)+1}(\x)\big),
\end{align}
where $j$ is uniformly distributed in the set $[2\mathscr{T}]$. Then, we can prove that this stochastic gradient function satisfies \assum{mss}.

\begin{lemma}\label{lem:mss-upper-bound}
The stochastic gradient function $\hat{g}$ defined in \eqn{quantum-hard-g-mss} is unbiased for $\nabla\tilde{f}_{T;U}(\x)$ and satisfies
\begin{align}
\E\big\|\hat{\g}(\x,j)-\nabla\tildef(\x)\big\|^2\leq\frac{4\mathscr{T}\alpha^2}{\beta^2},\qquad\E\big\|\hat{\g}(\x,j)-\hat{\g}(\y,j)\big\|^2\leq\frac{\hat{\ell}^2\mathscr{T}\alpha^2\|\x-\y\|^2}{\beta^2},\nonumber
\end{align}
for all $\x,\y\in\R^d$, where $\hat{\ell}=328$.
\end{lemma}
\begin{proof}
For any $\x$ and $j$, we define
\begin{align}
\delta(\x,j)\coloneqq\hat{\g}(\x,j)-\nabla \tildef(\x)=\Theta_{\prog_{\beta/2}(\x)+1}(\x)\cdot\big(\g_{\prog_{\beta/2}(\x)+1}(\x)-2\gamma\sqrt{\mathscr{T}}\cdot\vect{m}^{(j)}\big).\nonumber
\end{align}
Then we have
\begin{align}
\E\big\|\hat{\g}(\x,j)-\nabla\tildef(\x)\big\|^2
=\E\|\delta(\x,\z)\|^2\leq 4\mathscr{T}\alpha^2|\Theta_{\prog_{\beta/2}(\x)+1}(\x)|/\beta^2\leq 4\mathscr{T}\alpha^2/\beta^2.\nonumber
\end{align}
For any $\x,\y\in\R^d$,
\begin{align}
\hat{\g}(\x,j)-\hat{\g}(\y,j)=\delta(\x,j)-\delta(\y,j)+\nabla\tildef(\x)-\nabla\tildef(\y).\nonumber
\end{align}
Since $\E[\delta(\x,j)-\delta(\y,j)]=0$, we can derive that
\begin{align}
\E\big\|\hat{\g}(\x,j)-\hat{\g}(\y,j)\big\|^2=\E\big\|\delta(\x,j)-\delta(\y,j)\big\|^2+\|\nabla\tildef(\x)-\nabla\tildef(\y)\|^2.\nonumber
\end{align}
Note that
\begin{align}
\delta(\x,j)-\delta(\y,j)=\Theta_{i_\x}(\x)\cdot\big(\g_{i_\x}(\x)-2\gamma_\x\sqrt{\mathscr{T}}\cdot\vect{m}_\x^{(j)}\big)-\Theta_{i_\y}(\y)\cdot\big(\g_{i_\y}(\y)-2\gamma_\y\sqrt{\mathscr{T}}\cdot\vect{m}_\y^{(j)}\big),\nonumber
\end{align}
where we denote $i_\x=\prog_{\beta/2}(\x)+1$ and $i_\y=\prog_{\beta/2}(\y)+1$. Then,
\begin{align}
\E\big\|\delta(\x,j)-\delta(\y,j)\big\|^2
\leq &\, 2\mathscr{T}(\nabla_{i_\x}\tildef(\x))^2(\Theta_{i_\x}(\x)-\Theta_{i_\x}(\y))^2\nonumber\\
&+2\mathscr{T}\big(\nabla_{i_\x}\tildef(\x)-\nabla_{i_\x}\tildef(\y)\big)^2\Theta_{i_\x}^2(\y)\nonumber\\
&+2\mathscr{T}(\nabla_{i_\y}\tildef(\y))^2(\Theta_{i_\y}(\y)-\Theta_{i_\y}(\x))^2\nonumber\\
&+2\mathscr{T}\big(\nabla_{i_\y}\tildef(\y)-\nabla_{i_\y}\tildef(\x)\big)^2\Theta_{i_\y}^2(\x).\nonumber
\end{align}
As $\Theta_i$ is $36$-Lipschitz for any $i\in[T]$ according to \lem{Gamma-properties}, we have
\begin{align}
\E\big\|\delta(\x,j)-\delta(\y,j)\big\|^2
\leq &\,\mathscr{T}\cdot\big(2\alpha^2\cdot(23\cdot 6)^2\|\x-\y\|^2/\beta^2+2\|\nabla\tildef(\x)-\nabla\tildef(\y)\|^2\big)\nonumber\\
&+\|\nabla\tildef(\x)-\nabla\tildef(\y)\|^2\nonumber\\
\leq &\,\mathscr{T}\hat{\ell}^2\alpha^2\|\x-\y\|^2/\beta^2,\nonumber
\end{align}
where the last inequality uses the fact that the gradient of $\nabla\tildef$ is $152\alpha/\beta$-Lipschitz continuous, which is demonstrated in \lem{fT-boundedness}.
\end{proof}

Similar to the case of \sec{stochastic-construction-to-lowerbound}, we can show that if one only knows about the first $t$ components $\{\u^{(1)},\ldots,\u^{(t)}\}$, even if we permit the quantum algorithm to query the stochastic gradient oracle at different positions of $\x$, it is still hard to learn $\u^{(t+1)}$ as well as other components with larger indices. Quantitatively, following the same notation in \sec{stochastic-construction-to-lowerbound}, for any $1\leq t\leq T$ we denote
\begin{align}
W_{t;\perp}:=\Big\{\x\in\mathbb{B}(\0,\beta\sqrt{T})\,\big|\,\exists i,\text{ s.t. }|\<\x,\u^{(q)}\>|\geq\frac{\beta}{4}\text{ and }t<i\leq T\Big\},\nonumber
\end{align}
and
\begin{align}
W_{i;\parallel}:=\mathbb{B}(\0,\beta\sqrt{T})-W_{i;\perp},\nonumber
\end{align}
where $W_{t;\perp}$ is the subspace of $\mathbb{B}(\0,\beta\sqrt{T})$ such that any vector in $W_{t;\perp}$ has a relatively large overlap with at least one of $\u^{(t+1)},\ldots,\u^{(T)}$. Moreover, we still use $\Pi_{t;\perp}$ and $\Pi_{t;\parallel}$ to denote the quantum projection operators onto $W_{t;\perp}$ and $W_{t;\parallel}$, respectively. The following lemma demonstrates that, if starting in the subspace $W_{t;\parallel}$, any quantum algorithm using at most $\mathscr{T}/2$ queries at arbitrary locations cannot output a quantum state that has a large overlap with $W_{t;\perp}$ in expectation.

\begin{lemma}\label{lem:gradient-estimation-no-speedup-mss}
For any $n<\mathscr{T}/2$ and $t\leq T$, suppose in the form of \defn{quantum-SG-oracle} we are given the quantum stochastic gradient oracle $\tildeO_{\g;U}$ of $\g(\x,j)$ defined in Eq.~\eqn{quantum-hard-g-mss}. Then for any quantum algorithm $A_{\quan}$ in the form of Eq.~\eqn{quantum-algorithm-form}, consider the sequence of unitaries $A_{\quan}^{(n)}$ truncated after the $n$ stochastic gradient oracle query
\begin{align}
A_{\quan}^{(n)}:=\tildeO_{\g;U}V_n\tildeO_{\g;U}\cdots\tildeO_{\g;U}V_2\tildeO_{\g;U}V_1,\nonumber
\end{align}
and any input state $\ket{\phi}$, we have
\begin{align}\label{eqn:quantum-SGD-ineffective-mss}
\delta_{\perp}(n):=\E_{\{\u^{(t)},\ldots,\u^{(T)},M_\x\}}\big[\|\Pi_{t;\perp}\cdot A_{\quan}^{(n)}\ket{\phi}\|^2\big]\leq \frac{n}{18\mathscr{T}^2T^4},
\end{align}
where the expectation is over all possible sets $\{\u^{(t)},\ldots,\u^{(T)}\}$ and all possible sets of matrices $\{M_\x\}$ at all positions $\x\in\mathbb{B}(\0,\beta\sqrt{T})$ satisfy Eq.~\eqn{M-construction-mss}, given that the dimension $d$ of the objective function $\tilde{f}_{T;U}$ satisfies $d\geq2\mathscr{T}T\log\mathscr{T}$ and $\mathscr{T}\geq T$.
\end{lemma}

\begin{proof}
We use induction to prove this claim. First, for $n=1$, we have
\begin{align}
\delta_{\perp}(1)&=\E_{\{\u^{(t)},\ldots,\u^{(T)},M_\x\}}\big[\|\Pi_{t;\perp}\cdot \tildeO_{\g;U}V_0\ket{\phi}\|^2\big]\nonumber\\
&=\E_{\{\u^{(t)},\ldots,\u^{(T)},M_\x\}}\big[\|\Pi_{t;\perp}\cdot \tildeO_{\g;U}\ket{\phi}\|^2\big]\nonumber\\
&\leq\E_{\{\u^{(t)},\ldots,\u^{(T)},M_\x\}}\big[\big\|\Pi_{t;\perp}\cdot \tildeO_{\g;U}\ket{\phi_{\parallel}}\big\|^2\big]+\E_{\{\u^{(t)},\ldots,\u^{(T)},M_\x\}}\big[\big\|\ket{\phi_{\perp}}\big\|^2\big],\nonumber
\end{align}
where $\ket{\phi_{\parallel}}:=\Pi_{t;\parallel}\ket{\phi}$ and $\ket{\phi_{\perp}}:=\Pi_{t;\perp}\ket{\phi}$. Since for all components in the (possibly superposition) state $\Pi_{T;\perp}\ket{\psi}$ all the stochastic gradients have no overlap with $\{\u^{t+2},\ldots,\u^{T}\}$, by \lem{multivariate-mean-estimation} we have
\begin{align}
\E_{\{\u^{(t)},\ldots,\u^{(T)},M_\x\}}\big[\big\|\Pi_{t;\perp}\cdot \tildeO_{\g;U}\ket{\phi_{\parallel}}\big\|^2\big]\leq\exp(-\zeta \mathscr{T}),\nonumber
\end{align}
where $\zeta$ is a small enough constant. Moreover, by \lem{quantum-zero-chain} we have
\begin{align}
\E_{\{\u^{(t)},\ldots,\u^{(T)},M_\x\}}\big[\big\|\ket{\phi_{\perp}}\big\|^2\big]\leq \frac{1}{36\mathscr{T}^2T^4}.\nonumber
\end{align}
Hence,
\begin{align}
\delta_{\perp}(1)\leq\exp(-\zeta \mathscr{T})+\frac{1}{36\mathscr{T}^2T^4}\leq\frac{1}{18\mathscr{T}^2T^4}.\nonumber
\end{align}

Suppose the inequality \eqn{quantum-SGD-ineffective-mss} holds for all $n\leq\tilde{n}$ for some $\tilde{n}<\frac{\mathscr{T}}{2}$. Then for $n=\tilde{n}+1$, we denote
\begin{align}
\ket{\phi_{\tilde{n}}}:=\tildeO_{\g;U}V_{\tilde{n}-1}\tildeO_{\g;U}\cdots\tildeO_{\g;U}V_1\tildeO_{\g;U}V_0\ket{\phi}.\nonumber
\end{align}
Then,
\begin{align}
\delta_{\perp}(\tilde{n}+1)
&=\E_{\{\u^{(t)},\ldots,\u^{(T)},M_\x\}}\big[\|\Pi_{t;\perp}\cdot \tildeO_{\g;U}V_{\tilde{n}}\ket{\phi_{\tilde{n}}}\|^2\big]\nonumber\\
&\leq\E_{\{\u^{(t)},\ldots,\u^{(T)},M_\x\}}\big[\big\|\Pi_{t;\perp}\cdot \tildeO_{\g;U}V_{\tilde{n}}\ket{\phi_{\tilde{n};\parallel}}\big\|^2\big]
+\E_{\{\u^{(t)},\ldots,\u^{(T)},M_\x\}}\big[\big\|\ket{\phi_{\tilde{n};\perp}}\big\|^2\big]\nonumber\\
&\leq\E_{\{\u^{(t)},\ldots,\u^{(T)},M_\x\}}\big[\big\|\Pi_{t;\perp}\cdot\tildeO_{\g;U}V_{\tilde{n}}\ket{\phi_{\tilde{n};\parallel}}\big\|^2\big]
+\delta_{\perp}(\tilde{n}).\nonumber
\end{align}

Consider the following sequence
\begin{align}
\tildeO_{\g;U}V_{\tilde{n}}\tildeO_{\g;U}\cdots\tildeO_{\g;U}V_0\ket{\phi'}=\tildeO_{\g;U}V_{\tilde{n}}\ket{\phi_{\tilde{n};\parallel}},\nonumber
\end{align}
note that it contains $\tilde{n}+1\leq\frac{\mathscr{T}}{2}$ queries to the stochastic gradient oracle, and at each query except the last one, the input state has no overlap with the desired space $W_{t;\perp}$. Observe that within this restricted input subspace where these queries happen, we always have
\begin{align}
\mathbb{I}\{i>\prog_{\frac{\beta}{4}}(\x)\}=\Theta_i(\x)=\mathbb{I}\{i>\prog_{\frac{\beta}{2}}(\x)\}.\nonumber
\end{align}
Hence, the oracle behaves as if there is no scaling to the indicator function $\mathbb{I}\{i>\prog_{\frac{\beta}{4}}(\x)\}$, and we can apply \lem{multivariate-mean-estimation} to obtain the following result:
\begin{align}
\E_{\{\u^{(t)},\ldots,\u^{(T)},M_\x\}}\big[\big\|\Pi_{t;\perp}\cdot\tildeO_{\g;U}V_{\tilde{n}}\ket{\phi_{\tilde{n};\parallel}}\big\|^2\big]
&\leq \exp(-\zeta \mathscr{T})+\frac{1}{36\mathscr{T}^2T^4}\nonumber\\
&\leq \exp(-\zeta T)+\frac{1}{36\mathscr{T}^2T^4}\nonumber\\
&\leq \frac{1}{18\mathscr{T}^2T^4}.\nonumber
\end{align}
Hence, the inequality \eqn{quantum-SGD-ineffective-mss} also holds for $n=\tilde{n}+1$.
\end{proof}

%=======================================================================================================

\subsubsection{Lower Bound with Bounded Input Domain}

Through this construction of quantum stochastic gradient oracle with mean-squared smoothness, we can prove the query complexity lower bound for any quantum algorithm $A_{\quan}$ defined in \sec{quantum-model} using the hard instance $\tildef$ defined in Eq.~\eqn{tildef-defn}. For the convenience of notations, we use $\widetilde{O}_{\g;U}$ to denote the stochastic gradient oracle defined in Eq.~\eqn{quantum-hard-g-mss} of function $\tildef$. Similar to \sec{noiseless-lowerbound} and \append{bounded-lower}, we consider the truncated sequence $A_{\quan}^{(K\cdot \mathscr{T}/2)}$ of any possible quantum algorithm $A_{\quan}$ with $K<T$, and define a sequence of unitaries starting with $A_0=A_{\quan}^{(K\cdot \mathscr{T}/2)}$ as follows:
\begin{align}\label{eqn:stochastic-unitary-sequences-mss}
A_0&:=V_{K+1}\tildeO_{\g;U}V_{K;\mathscr{T}/2}\cdots\tildeO_{\g;U}V_{K;1}\cdots\tildeO_{\g;U}V_{2;\mathscr{T}/2}\cdots\tildeO_{\g;U}V_{2;1}\tildeO_{\g;U}V_{1;\mathscr{T}/2}\cdots\tildeO_{\g;U}V_{1;1}\\ \nonumber
A_1&:=V_{K+1}\tildeO_{\g;U}V_{K;\mathscr{T}/2}\cdots\tildeO_{\g;U}V_{K;1}\cdots\tildeO_{\g;U}V_{2;\mathscr{T}/2}\cdots\tildeO_{\g;U}V_{2;1}\tildeO_{\g;U_1}V_{1;\mathscr{T}/2}\cdots\tildeO_{\g;U_1}V_{1;1}\\ \nonumber
A_2&:=V_{K+1}\tildeO_{\g;U}V_{K;\mathscr{T}/2}\cdots\tildeO_{\g;U}V_{K;1}\cdots\tildeO_{\g;U_2}V_{2;\mathscr{T}/2}\cdots\tildeO_{\g;U_2}V_{2;1}\tildeO_{\g;U_1}V_{1;\mathscr{T}/2}\cdots\tildeO_{\g;U_1}V_{1;1}\\ \nonumber
&\vdots\\ \nonumber
A_K&:=V_{K+1}\tildeO_{\g;U_K}V_{K;\mathscr{T}/2}\cdots\tildeO_{\g;U_K}V_{K;1}\cdots\tildeO_{\g;U_2}V_{2;\mathscr{T}/2}\cdots\tildeO_{\g;U_2}V_{2;1}\tildeO_{\g;U_1}V_{1;\mathscr{T}/2}\cdots\tildeO_{\g;U_1}V_{1;1},
\end{align}
where $\tildeO_{\g;U_t}$ stands for the stochastic gradient oracle of the function $\tilde{f}_{t;U_t}$. Note that for the sequence of unitaries $A_0$, it can be decomposed into the product of $V_{K+1}$ and $K$ unitaries, each of the form
\begin{align}
\mathscr{A}_k(n)=\tildeO_{\g;U}V_{k;n}\tildeO_{\g;U}\cdots\tildeO_{\g;U}V_{k;2}\tildeO_{\g;U}V_{k;1}\nonumber
\end{align}
for $n=\mathscr{T}/2$ and $k\in[K]$ for some unitaries $V_1,\ldots,V_n$. In the following lemma, we demonstrate that for such a sequence $\mathscr{A}_k(n)$, if we replace $\tilde{O}_{\g;U}$ by another oracle that only reveals part information of $f$, the sequence will barely change on random inputs.

\begin{lemma}\label{lem:part-sequences-closeness-mss}
For any $t\in[T-1]$ and any $n\leq\frac{\mathscr{T}}{2}$, consider the following two sequences of unitaries
\begin{align}
\mathscr{A}(n)=\tildeO_{\g;U}V_n\tildeO_{\g;U}\cdots\tildeO_{\g;U}V_2\tildeO_{\g;U}V_1,\nonumber
\end{align}
and
\begin{align}
\hat{\mathscr{A}}_t(n)=\tildeO_{\g;U_t}V_n\tildeO_{\g;U_t}\cdots\tildeO_{\g;U_t}V_2\tildeO_{\g;U_t}V_1,\nonumber
\end{align}
we have
\begin{align}\label{eqn:SGD-closeness-mss}
\delta(n):=\E_{\{\u^{(t)},\ldots,\u^{(T)},M_\x\}}\big[\|(\hat{\mathscr{A}}_t(n)-\mathscr{A}(n))\ket{\psi}\|^2\big]\leq\frac{n}{36\mathscr{T}T^4}
\end{align}
for any pure state $\ket{\psi}$.
\end{lemma}
\begin{proof}
We use induction to prove this claim. First, for $n=1$, we have
\begin{align}
&\E_{\{\u^{(t)},\ldots,\u^{(T)},M_\x\}}\big[\big\|(\hat{\mathscr{A}}_t(n)-\mathscr{A}(n))\ket{\psi}\big\|^2\big]\nonumber\\
&\qquad=\E_{\{\u^{(t)},\ldots,\u^{(T)},M_\x\}}\big[\big\|(\tildeO_{\g;U}-\tildeO_{\g;U_t})\ket{\psi}\big\|^2\big]\leq \frac{1}{36\mathscr{T}^2T^4},\nonumber
\end{align}
where the last inequality follows from \lem{quantum-zero-chain}. Suppose the inequality \eqn{SGD-closeness} holds for all $n\leq\tilde{n}$ for some $\tilde{n}<\frac{T}{2}$. Then for $n=\tilde{n}+1$, we have
\begin{align}
\delta(\tilde{n}+1)
&=\E_{\{\u^{(t)},\ldots,\u^{(T)},M_\x\}}\big[\|(\hat{\mathscr{A}}_t(n)-\mathscr{A}(n))\ket{\psi}\|^2\big]\\
&\leq\E_{\{\u^{(t)},\ldots,\u^{(T)},M_\x\}}\big[\|(\tildeO_{\g;U}-\tildeO_{\g;U_t})\ket{\psi_t}\|^2\big]+\delta(\tilde{n}),
\end{align}
where
\begin{align}
\ket{\psi_t}=V_{\tilde{n}}\tildeO_{\g;U_t}\cdots\tildeO_{\g;U_t}V_1\ket{\psi}
\end{align}
is a function of $U_t$ obtained by $\tilde{n}$ queries to $\tilde{O}_{\g;U_t}$. By \lem{gradient-estimation-no-speedup}, we have
\begin{align}
\E_{\{\u^{(t)},\ldots,\u^{(T)},M_x\}}[\|\Pi_{t;\perp}\ket{\psi_t}\|^2]\leq \frac{n}{18\mathscr{T}^2T^4}\leq\frac{1}{36\mathscr{T}T^4},
\end{align}
indicating that $\ket{\psi_t}$ only has a very little overlap with the subspace $W_{t;\perp}$ defined in \eqn{W_t_perp-defn}, outside of which the columns $\{\u^{(t)},\ldots,\u^{(T)}\}$ of $U$ has no impact on the function value and derivatives of $\tilde{f}_{T;U}$. Thus,
\begin{align}
\delta(\tilde{n}+1)
&\leq\E_{\{\u^{(t)},\ldots,\u^{(T)},M_\x\}}\big[\|(\tildeO_{\g;U}-\tildeO_{\g;U_t})\ket{\psi_t}\|\big]+\delta(\tilde{n})\nonumber\\
&\leq\E_{\{\u^{(t)},\ldots,\u^{(T)},M_x\}}[\|\Pi_{t;\perp}\ket{\psi_t}\|^2]+\delta(\tilde{n})\leq\frac{\tilde{n}+1}{36\mathscr{T}T^4},\nonumber
\end{align}
indicating that Eq.~\eqn{SGD-closeness-mss} also holds for $n=\tilde{n}+1$.
\end{proof}

\begin{lemma}[$A_t$ and $A_{t-1}$ have similar outputs]
\label{lem:similar-outputs-stochastic-mss}
For a hard instance $\tildef(\x)\colon\R^d\to\R$ defined on $\mathbb{B}(\0,2\beta\sqrt{T})$ with $d\geq 2\mathscr{T}T\log\mathscr{T}$, let $A_t$ for $t\in[K]$ be the sequence unitaries defined in Eq.~\eqn{stochastic-unitary-sequences}. Then
\begin{align}
\E_{\{\u^{(t)},\ldots,\u^{(T)},M_\x\}}\big(\|A_t\ket{\0}-A_{t-1}\ket{\0}\|^2\big)\leq\frac{1}{72T^4}.\nonumber
\end{align}
\end{lemma}

\begin{proof}
From the definition of the unitaries in Eq.~\eqn{stochastic-unitary-sequences-mss}, we have
\begin{align}
\|A_t\ket{\0}-A_{t-1}\ket{\0}\|=\|(\mathscr{A}(\mathscr{T}/2)-\hat{\mathscr{A}}_t(\mathscr{T}/2))\ket{\psi}\|\nonumber
\end{align}
for some fixed quantum state $\ket{\psi}$ dependent on the vectors $\{\u^{(1)},\ldots,\u^{(t-1)}\}$, where
\begin{align}
\mathscr{A}(\mathscr{T}/2)=\tildeO_{\g;U}V_{\mathscr{T}/2}\tildeO_{\g;U}\cdots\tildeO_{\g;U}V_2\tildeO_{\g;U}V_1,\nonumber
\end{align}
and
\begin{align}
\hat{\mathscr{A}}_t(\mathscr{T}/2)=\tildeO_{\g;U_t}V_{\mathscr{T}/2}\tildeO_{\g;U_t}\cdots\tildeO_{\g;U_t}V_2\tildeO_{\g;U_t}V_1.\nonumber
\end{align}
By \lem{part-sequences-closeness-mss}, we have
\begin{align}
\E_{\{\u^{(t)},\ldots,\u^{(T)},M_{\x}\}}\big(\|(\mathscr{A}(T/2)-\hat{\mathscr{A}}_t(T/2))\ket{\psi}\|^2\big)\leq\frac{1}{36\mathscr{T}T^4}\cdot\frac{\mathscr{T}}{2}=\frac{1}{72T^4},\nonumber
\end{align}
which leads to
\begin{align}
\E_{\{\u^{(t)},\ldots,\u^{(T)},M_\x\}}\big(\|A_t\ket{\0}-A_{t-1}\ket{\0}\|^2\big)\leq\frac{1}{72T^4}.\nonumber
\end{align}
\end{proof}

\begin{proposition}\label{prop:stochastic-A_0-cannot-mss}
Consider the $d$-dimensional function $\tildef(\x)\colon\mathbb{B}(\0,2\beta\sqrt{T})\to\R$ defined in \eqn{tildef-defn} with the rotation matrix $U$ being chosen arbitrarily and the dimension $d\geq 2\mathscr{T}T\log\mathscr{T}$ and $\mathcal{T}\geq T$. Consider the truncated sequence $A_{\quan}^{(K\mathscr{T}/2)}$ of any possible quantum algorithm $A_{\quan}$ containing $K\mathscr{T}/2$ queries to the quantum stochastic gradient oracle $\tildeO_{\g;U}$ of $\g(\x,j)$ defined in Eq.~\eqn{quantum-hard-g-mss} with $K<T$, and let $p_U$ be the probability distribution over $\x\in\mathbb{B}(\0,2\beta\sqrt{T})$ obtained by measuring the state $A_{\quan}^{(K\mathscr{T}/2)}\ket{0}$, which is related to the rotation matrix $U$. Then,
\begin{align}
\Pr_{U,\{M_\x\},\x_{\text{out}}\sim p_U
}\big[\|\nabla\tildef(\x_{\text{out}})\|\leq\alpha/\beta\big]\leq\frac{1}{3},\nonumber
\end{align}
where the probability is subject to all possible orthogonal rotation matrices $U$, and all possible matrices $\{M_{\x}\}$ in the quantum stochastic gradient function $\g(\x,j)$ for any $\x$.
\end{proposition}

\begin{proof}
We first demonstrate that the sequence of unitaries $A_K$ defined in Eq.~\eqn{stochastic-unitary-sequences-mss} cannot find an $\alpha/\beta$-approximate stationary point with high probability. In particular, let $p_{U_K}$ be the probability distribution over $\x\in\mathbb{B}(\0,2\beta\sqrt{T})$ obtained by measuring the output state $A_K\ket{0}$. Then,
\begin{align}
\Pr_{U_K,\{M_\x\},\x_{\text{out}}\in p_{U_K}}\big[\|\nabla\tildef(\x_{\text{out}})\|\leq \alpha/\beta\big]
\leq \Pr_{\{\u^{(k+1)},\ldots,\u^{(T)}\}}\big[\|\nabla\tildef(\x)\|\leq \alpha/\beta\big].\nonumber
\end{align}
for any fixed $\x$. Then by \lem{cannot-guess} we have
\begin{align}
\Pr_{U_K,{M_\x},\x_{\text{out}}\in p_{U_K}}\big[\|\nabla\tildef(\x_{\text{out}})\|\leq\alpha/\beta\big]\leq \frac{1}{6}.\nonumber
\end{align}
Moreover, by \lem{similar-outputs-stochastic-mss} and Cauchy-Schwarz inequality, we have
\begin{align}
\E_{U}\big[\|A_K\ket{0}-A_0\ket{0}\|^2\big]\leq K\cdot\E_U\Big[\sum_{t=1}^{K-1}\|A_{t+1}\ket{0}-A_t\ket{0}\|^2\Big]\leq\frac{1}{72T^2}.\nonumber
\end{align}
Then by Markov's inequality,
\begin{align}
\Pr_{U}\Big[\|A_{K-1}\ket{0}-A_0\ket{0}\|^2\geq\frac{1}{6T}\Big]\leq\frac{1}{6T},\nonumber
\end{align}
since both norms are at most 1. Thus, the total variance distance between the probability distribution $p_U$ obtained by measuring $A_0\ket{0}$ and the probability distribution $p_{U_K}$ obtained by measuring $A_K\ket{0}$ is at most
\begin{align}
\frac{1}{6T}+\frac{1}{6T}=\frac{1}{3T}\leq\frac{1}{6}.\nonumber
\end{align}
Hence, we can conclude that
\begin{align}
\hspace{-2mm}\Pr_{U,\{M_{\x}\},\x_{\text{out}}\sim p_U^{(t)}
}\big[\|\nabla\tildef(\x_{\text{out}})\|\leq\alpha/\beta\big]
\leq\Pr_{U_K,\{M_{\x}\},\x_{\text{out}}\sim p_{U_K}}\big[\|\nabla\tildef(\x_{\text{out}})\|\leq\alpha/\beta\big]+\frac{1}{6}=\frac{1}{3}.\nonumber
\end{align}
\end{proof}

\begin{proposition}\label{prop:stochastic-bounded-mss}
Suppose $\Delta$, $\bar{L}$, $\sigma$, and $\epsilon$ are positive. Then,
\begin{align}
\mathcal{T}^{\sto}_{\epsilon}\big(\mathcal{A}_{\quan},\tilde{\mathcal{F}}_1(\Delta,\bar{L},\mathcal{R}),\sigma\big)=\Omega\Big(\frac{\Delta\bar{L}\sigma}{\epsilon^3}\Big),\nonumber
\end{align}
if we further assume the stochastic gradient function $\g(\x)$ satisfies \assum{mss} with mean-squared smoothness parameter $\bar{L}$, where $\mathcal{R}=\sqrt{\frac{\hat{\ell}\sigma\Delta}{6\ell\bar{L}\gamma\epsilon}}$, the complexity measure $\mathcal{T}^{\sto}_{\epsilon}(\cdot)$ is defined in Eq.~\eqn{stochastic-complexity-measure}, and the function class $\tilde{\mathcal{F}}_1(\Delta,\bar{L},\mathcal{R})$ is defined in \defn{tildeFp-bounded}. The lower bound holds even if we restrict $\tilde{\mathcal{F}}_1(\Delta,\bar{L},\mathcal{R})$ to functions whose domain has dimension
\begin{align}
\tilde{\Theta}\Big(\frac{\Delta\bar{L}\sigma}{\epsilon^3}\Big).\nonumber
\end{align}
\end{proposition}

\begin{proof}
We set up the scaling parameters $\alpha$ and $\beta$ in the hard instance $\tildef\colon\R^d\to\R$ defined in Eq.~\eqn{tildef-defn} as
\begin{align}
\alpha=\frac{L\beta^{2}}{\ell},\qquad\beta=\frac{2\ell\epsilon}{L},\nonumber
\end{align}
where $\ell$ is the gradient Lipschitz constant of $\bar{f}_T$ whose value is given in \lem{fT-boundedness}, and the parameter $L\leq\bar{L}$ is specified later. We also set the parameters
\begin{align}
T=\frac{L\Delta}{48\ell\epsilon^2},\qquad\mathscr{T}=\frac{\sigma^2\beta^2}{4\gamma^2\alpha^2}=\frac{\sigma^2}{4\gamma^2\epsilon^2}.\nonumber
\end{align}
Then by \lem{fT-boundedness}, we know that $\tildef$ is $L$-smooth and thus $\bar{L}$-smooth since $L\leq\bar{L}$, and
\begin{align}
\tildef(\0)-\inf_{\x}\tildef(\x)=\alpha\big(\bar{f}_T(\0)-\inf_\x\bar{f}_T(\x)\big)\leq\frac{12L\beta^2}{\ell}\cdot T\leq\Delta,\nonumber
\end{align}
indicating that $\tildef\in\tilde{\mathcal{F}}(\Delta,L_p,\mathcal{R})$ for arbitrary dimension $d$ and rotation matrix $U$. Moreover, for every $\x,\y\in\R^d$, by \lem{mss-upper-bound} we have
\begin{align}
\mathbb{E}\big[\|\g(\x,j)-\nabla\tildef(\x)\|^2\big]\leq 4\alpha^2\gamma^2\mathscr{T}/\beta^2\leq\delta^2,\nonumber
\end{align}
indicating that the variance of the stochastic gradient function $\g$ defined in Eq.~\eqn{quantum-hard-g-mss} is bounded by $\sigma^2$, and
\begin{align}
\E\big\|\hat{\g}(\x,j)-\hat{\g}(\y,j)\big\|^2\leq\frac{\hat{\ell}^2\mathscr{T}\alpha^2\|\x-\y\|^2}{\beta^2}=\frac{\hat{\ell}^2L^2\mathscr{T}}{\ell^2}\cdot\|\x-\y\|^2.\nonumber
\end{align}
Hence, to guarantee that \assum{mss} is satisfied, we set
\begin{align}
L=\frac{\ell}{\hat{\ell}\sqrt{\mathscr{T}}}\cdot\bar{L}=\frac{2\ell\gamma\epsilon}{\hat{\ell}\sigma}\cdot\bar{L}.\nonumber
\end{align}
Furthermore, we notice that the radius $\mathcal{R}$ satisfies
\begin{align}
\mathcal{R}=2\beta \sqrt{T}=\frac{4\ell\epsilon}{L}\cdot\sqrt{\frac{L\Delta}{48\ell\epsilon^2}}=\sqrt{\frac{\ell\Delta}{3L}}=\sqrt{\frac{\hat{\ell}\sigma\Delta}{6\ell\bar{L}\gamma\epsilon}}.\nonumber
\end{align}
To guarantee that $\mathscr{T}\geq T$, $\epsilon$ has to satisfy
\begin{align}
\frac{\sigma^2}{4\gamma^2\epsilon^2}\geq\frac{\Delta}{48\ell\epsilon^2}\cdot\frac{2\ell\gamma\epsilon\bar{L}}{\hat{\ell}\sigma},\nonumber
\end{align}
indicating
\begin{align}
\epsilon\leq\frac{\sigma^2}{4\gamma^2}\cdot\frac{24\hat{\ell}\sigma}{\Delta\gamma\bar{L}}=\frac{6\sigma^2\hat{\ell}}{4\gamma^3\Delta\bar{L}}.\nonumber
\end{align}

By \prop{stochastic-A_0-cannot}, for any truncated sequence $A_{\quan}^{(K\mathscr{T}/2)}$ of any possible quantum algorithm $A_{\quan}$ containing $K\mathscr{T}/2<T\mathscr{T}/2$ queries to the oracle $O^{(p)}_f$ on input domain $\mathbb{B}(0,\mathcal{R})$, we have
\begin{align}
\Pr_{U,\x_{\text{out}}\sim p_U
}\big[\|\nabla\tildef(\x_{\text{out}})\|\leq\alpha/\beta\big]=\Pr_{U,\x_{\text{out}}\sim p_U
}\big[\|\nabla\tildef(\x_{\text{out}})\|\leq\epsilon\big]\leq\frac{1}{3},\nonumber
\end{align}
where $p_U$ is the probability distribution over $\x\in\mathbb{B}(\0,2\beta\sqrt{T})=\mathbb{B}(\0,\mathcal{R})$ obtained by measuring the state $A_{\quan}^{(K\mathscr{T}/2)}\ket{0}$, given that the dimension $d$ satisfies
\begin{align}
d\geq 2\mathscr{T}T\log\mathscr{T}=\tilde{\Theta}\Big(\frac{\Delta\bar{L}\sigma}{\epsilon^3}\Big).\nonumber
\end{align}
Then according to \defn{quantum-complexity-measure} we can conclude that
\begin{align}
\mathcal{T}_{\epsilon}^{\sto}\big(\mathcal{A}_{\quan},\tilde{\mathcal{F}}_1(\Delta,\bar{L},\mathcal{R},\sigma)\big)
\geq \frac{\mathscr{T}T}{2}
=\tilde{\Omega}\Big(\frac{\Delta\bar{L}\sigma}{\epsilon^3}\Big).\nonumber
\end{align}
\end{proof}

\subsubsection{Lower Bound with Unbounded Input Domain}\label{append:stochastic-lowerbound-unbounded-mss}
In this subsection, we extend the quantum lower bound proved in \prop{stochastic-bounded-mss} to the function class $\mathcal{F}(\Delta,\bar{L})$ with unbounded input domain via similar scaling techniques adopted in~\citet{arjevani2022lower}, \append{thm-p-th-order-formal}, and \append{stochastic-lowerbound-unbounded}. In particular, we consider the scaled hard instance $\hatf$ introduced in~\citet{carmon2020lower} and also used in~\citet{arjevani2022lower},
\begin{align}
\hatf(\x):=\tildef(\chi(\x))+\frac{\alpha}{10}\cdot\frac{\|\x\|^2}{\beta^2},\nonumber
\end{align}
where
\begin{align}
\chi(\x):=\frac{\x}{\sqrt{1+\|\x\|^2/\hat{\mathcal{R}}^2}},\nonumber
\end{align}
with the following parameters
\begin{align}
\alpha=\frac{L\beta^2}{\ell}\qquad \beta=\frac{2\ell\epsilon}{L},\qquad T=\frac{L\Delta}{48\ell\epsilon},\qquad L=\frac{2\ell\gamma\epsilon\bar{L}}{\hat{\ell}\sigma},\qquad\hat{\mathcal{R}}=230\beta\sqrt{T},\nonumber
\end{align}
whose values are also adopted in the proof of \prop{stochastic-bounded-mss}. The constants in $\hatf$ are chosen carefully such that stationary points of $\hatf$ are in one-to-one correspondence to stationary points of the hard instance $\tildef$ concerning the setting with bounded input domain. Quantitatively,

\begin{lemma}[\citealt{arjevani2022lower}, Section 4]\label{lem:stochastic-tilde-hat-correspondence-mss}
Let $\Delta$, $\bar{L}$, and $\epsilon$ be positive constants. Then, under the following choice of parameters
\begin{align}
\alpha=\frac{L\beta^2}{\ell},\qquad\beta=\frac{2\ell\epsilon}{L},\qquad T=\frac{L\Delta}{48\ell\epsilon},\qquad L=\frac{2\ell\gamma\epsilon\bar{L}}{\hat{\ell}\sigma},\nonumber
\end{align}
where $\ell$ is the gradient Lipschitz parameter of $\bar{f}_T$ whose value is given in \lem{fT-boundedness}, such that for any function pairs $(\tildef,\hatf)\in\tilde{\mathcal{F}}_1(\Delta,\bar{L},\mathcal{R})\times\mathcal{F}_1(\Delta,\bar{L})$ with dimension $d\geq 400T\log T$ and the same rotation matrix $U$, there exists a bijection between the $\epsilon$-stationary points of $\tildef$ and the $\epsilon$-stationary points of $\hatf$ that is independent from $U$.
\end{lemma}
Equipped with \lem{stochastic-tilde-hat-correspondence-mss}, we are now ready to prove \thm{stochastic-formal-mss}.

\begin{proof}[Proof of \thm{stochastic-formal-mss}]
Note that one quantum query to the stochastic gradient of $\hatf$ can be implemented by one quantum query to the stochastic gradient of $\tildef$ with the same rotation $U$, if we directly scale the stochastic gradient function of $\tildef$ to $\hatf$, which will not increase the variance of the stochastic gradient function, and the mean-squared smoothness condition in \assum{mss} is still preserved with the same mean-squared smoothness parameter $\bar{L}$. Combined with \lem{stochastic-tilde-hat-correspondence-mss}, we can note that the problem of finding $\epsilon$-stationary points of $\tildef$ with unknown $U$ can be reduced to the problem of finding $\epsilon$-stationary points of $\hatf$ with no additional overhead in terms of query complexity. Then by \prop{stochastic-bounded-mss}, we can conclude that
\begin{align}
\mathcal{T}^{\sto}_{\epsilon}\big(\mathcal{A}_{\quan},\mathcal{F}_1(\Delta,\bar{L}),\sigma\big)\geq \mathcal{T}^{\sto}_{\epsilon}\big(\mathcal{A}_{\quan},\tilde{\mathcal{F}}_1(\Delta,\bar{L},\mathcal{R}),\sigma\big)=\Omega\Big(\frac{\Delta\bar{L}\sigma}{\epsilon^3}\Big),\nonumber
\end{align}
if we further assume the stochastic gradient function satisfies \assum{mss}. The dimension dependence is the same as \prop{stochastic-bounded-mss}.
\end{proof}

%%%%%%%%%%%%%%%%%%%%%%%%%%%%%%%%%%%%%%%%%%%%%%%%%%%%%%%%%%%%%%%%%%%%%%%%%%%%%%%

\end{document}